\documentclass[a4paper,twoside, 11pt]{journal}
\usepackage{fullpage}
\usepackage[usenames,dvipsnames]{xcolor}
\usepackage{amssymb, amsmath, amsthm, wasysym, mathrsfs}
\usepackage{xspace}
\usepackage{graphicx, wrapfig, subfig}
\usepackage[colorlinks]{hyperref}
\usepackage[english]{babel}

\usepackage{enumitem}
\usepackage{cite}
\usepackage{caption}
\usepackage{array}
\usepackage{mathtools}
\usepackage{titlesec}

\usepackage{times}

\usepackage{microtype}

\graphicspath{{figures/}}

\hypersetup{allcolors = violet}

\author{
  Maarten L\"offler\thanks{
    Department of Information and Computing Sciences, Utrecht University,
    \texttt{\{m.loffler,f.staals\}@uu.nl}
  }
  \and
  Martin N\"ollenburg\thanks{
    Institute of Theoretical Informatics, Karlsruhe Institute of Technology (KIT),
    \texttt{noellenburg@kit.edu}
  }
  \and
  Frank Staals\footnotemark[1]
}
\title{Mixed Map Labeling}

\newcommand{\myremark}[4]{\textcolor{blue}{\textsc{#1 #2:}} \textcolor{#4}{\textsf{#3}}}
\renewcommand{\myremark}[4]{}

\newcommand{\frank}[2][says]{\myremark{Frank}{#1}{#2}{SeaGreen}}
\newcommand{\maarten}[2][says]{\myremark{Maarten}{#1}{#2}{Purple}}
\newcommand{\martin}[2][says]{\myremark{Martin}{#1}{#2}{Maroon}}

\newtheorem{theorem} {Theorem}
\newtheorem{proposition}[theorem] {Proposition}
\newtheorem{lemma}[theorem] {Lemma}
\newtheorem{corollary}[theorem] {Corollary}

\newtheorem{observation}[theorem] {Observation}

\DeclareFontFamily{OT1}{pzc}{}
\DeclareFontShape{OT1}{pzc}{m}{it}{<-> s * [0.900] pzcmi7t}{}
\DeclareMathAlphabet{\mathpzc}{OT1}{pzc}{m}{it}

\newcommand{\mkmcal}[1]{\ensuremath{\mathcal{#1}}\xspace}

\newcommand{\LL}{\mkmcal{L}}
\newcommand{\T}{\mkmcal{T}}
\newcommand{\C}{\mkmcal{C}}

\newcommand{\U}{\mkmcal{U}}
\newcommand{\D}{\mkmcal{D}}
\newcommand{\NP}{\textsf{NP}}

\newcommand{\M}{\mkmcal{M}}
\newcommand{\E}{\mkmcal{E}}
\newcommand{\I}{\mkmcal{I}}

\renewcommand{\P}{\mkmcal{P}}

\newcommand{\lbl}[1]{\ensuremath{\lambda_{#1}}\xspace}
\newcommand{\ldr}[1]{\ensuremath{\gamma_{#1}}\xspace}
\newcommand{\lab}[1]{\ensuremath{\Phi(#1)}\xspace}

\makeatletter
\renewcommand*{\@fnsymbol}[1]{\ensuremath{\ifcase#1\or *\or \mathsection\or \mathparagraph\or
   \dagger\or \ddagger\or \|\or **\or \dagger\dagger
   \or \ddagger\ddagger \else\@ctrerr\fi}}
\makeatother

\titleformat{\paragraph}[runin]{\bfseries}{\theparagraph}{0}{}[.]

\begin{document}
\maketitle

\begin{abstract}
  Point feature map labeling is a geometric problem, in which a set of input
  points must be labeled with a set of disjoint rectangles (the bounding boxes
  of the label texts). Typically, labeling models either use internal labels,
  which must touch their feature point, or external (boundary) labels, which
  are placed on one of the four sides of the input points' bounding box and
  which are connected to their feature points by crossing-free leader lines. In
  this paper we study polynomial-time algorithms for maximizing the number of
  internal labels in a mixed labeling model that combines internal and external
  labels. The model requires that all leaders are parallel to a given
  orientation $\theta \in [0,2\pi)$, whose value influences the geometric
  properties and hence the running times of our algorithms.
\end{abstract}

\section{Introduction}
\label{sec:Introduction}

Annotating features of interest in information graphics with textual labels or
icons is an important task in information visualization. One classical
application, whose principles easily generalize to the labeling of other
illustrations, is map labeling, where labels are mostly placed internally in
the map. Common cartographic placement guidelines demand that each label is
placed in the immediate neighborhood of its feature and that the association
between labels and features is unambiguous, while no two labels may overlap
each other~\cite{i-pnm-75,rr-plhcqmmwpi-14}. Point feature labeling has been
studied extensively in the computational geometry literature, but also in the
application areas. It is known that maximizing the number of non-overlapping
labels for a given set of input points is \NP-hard, even for very restricted
labeling models~\cite{fw-ppwalm-91,ms-ccclp-91}. In terms of labeling
algorithms, several approximations, polynomial-time approximation schemes
(PTAS), and exact approaches are
known~\cite{aks-lpmir-98,ksw-plwsl-99,cc-mir-09,km-olpfrlm-03}, as well as many
practically effective heuristics, see the bibliography of Wolff and
Strijk~\cite{ws-mlb-09}. If, however, feature points lie too dense in the map
or if their labels are relatively large, often only small fractions of the
features obtain a label, even in an optimal solution.

An alternative labeling approach using external instead of internal labels is
known as \emph{boundary labeling} in the literature. This labeling style is
frequently used when annotating anatomical drawings and technical
illustrations, where different and often small parts are identified using
labels and descriptive texts outside the actual picture, which are connected to
their features using leader lines. While the association between points and
external labels is often more difficult to see, the big advantages of boundary
labeling are that even dense feature sets can be labeled and that larger labels
can be accommodated on the margins of the illustration. Many efficient boundary
labeling algorithms are known. They can be classified by the leader shapes
that are used and by the sides of the picture's bounding box that are used for
placing the
labels~\cite{bksw-blmearm-07,bhkn-amcbl-09,bkns-blol-10,ghn-bapi-11,knrssw-tblwa
s-13,nps-dosbl-10,hpl-blflp-14}.

The combination of internal and external labeling models using internal labels
where possible and external labels where necessary seems natural and has been proposed as an open problem by Kaufmann~\cite{k-lwl-09}; however, only
few results are known in such \emph{hybrid} or \emph{mixed} settings.
In a mixed labeling, the final image will be an overlay of the original map or drawing, 
a collection of internal text labels, and a collection of leaders leading out of the image.
Depending on the application, we may wish to forbid intersections between some or all of these layers.
When no additional intersection constraints are imposed, the problem reduces to classical 
internal map labeling. Under the very natural restriction that leaders cannot
intersect internal labels, the
internal labels need to be placed carefully, as not every set of disjoint
internal labels creates sufficient gaps for routing leaders of the prescribed
shape from all remaining feature points to the image boundary. L\"offler and
N\"ollenburg~\cite{ln-sbolb-10} studied a restricted case of hybrid labeling,
where a partition of the feature points into points with internal fixed
position labels and points with external labels to be connected by one-bend
orthogonal leaders is given as input. They presented efficient algorithms and
hardness results, depending on three different problem parameters. Bekos et
al.~\cite{bkps-ctlwbl-11} studied a mixed labeling model with fixed-position
internal labels and external labels on one or two opposite sides of the
bounding box, connected by two-bend orthogonal leaders. Their goal is to
maximize the number of internally labeled points, while labeling all remaining
points externally. Polynomial and quasi-polynomial-time algorithms, as well as
an approximation algorithm and an ILP formulation were presented. 

\begin{figure}[tb]
	\centering
	\subfloat[$\theta = \pi$]{\includegraphics[scale=0.52,page=3]{motivation-example}}
	\hfill
	\subfloat[$\theta=0$]{\includegraphics[scale=0.5,page=4]{motivation-example}}
	\hfill
	\subfloat[\label{sfg:expl_nice}$\theta = \pi/3$]{\includegraphics[scale=0.5,page=1]{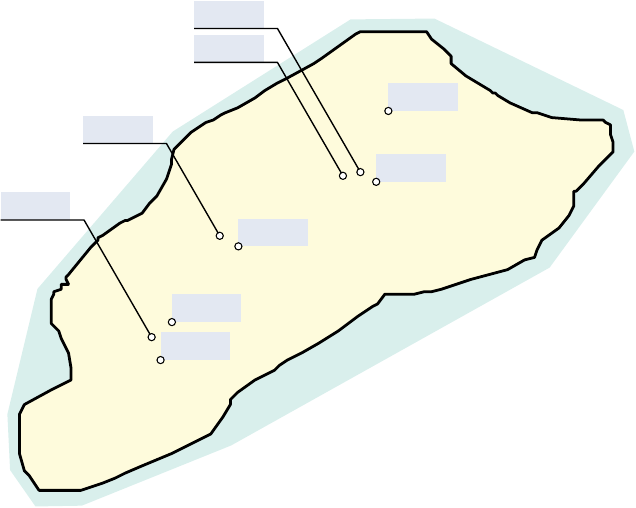}}
	\caption{A sample point set with mixed labelings of three different slopes. In (a) five external labels are necessary, whereas (b) and (c) require only four external labels. The slope in (c) yields aesthetically pleasing results.}
	\label{fig:example}
\end{figure}

\paragraph{Contribution} In this paper, we extend the known results on mixed
map labeling as follows. We present a mixed labeling model, in which each point
is assigned either an axis-aligned fixed-position internal label (e.g., to the
top right of the point) or an external label connected with a leader of slope
$\theta$, where $\theta \in [0,2\pi)$ is an input parameter defining the unique
leader direction for all external labels, measured clockwise from the negative
$x$-axis (see Fig.~\ref{fig:example}). In this model, we present a new
dynamic-programming algorithm to maximize the number of internally labeled
points for any given slope $\theta$, including the left- and right-sided
case ($\theta=0$ or $\theta=\pi$), which was studied by Bekos et
al.~\cite{bkps-ctlwbl-11}. While for the right-sided case Bekos et al.\ provided a faster $O(n
\log^2 n)$-time algorithm, where $n$ is the number of input points, our
algorithm improves upon their pseudo-polynomial $O(n^{\log n+3})$-time
algorithm for the left-sided case. We solve this problem in $O(n^3(\log n +
\delta))$ time, where $\delta = \min\{n, 1/d_{\min}\}$ is the inverse of the
distance $d_{\min}$ of the closest pair of points in \P and expresses the
maximum density of \P (Section~\ref{sec:left_leaders}). In the general case it turns out that
the set of slopes can be partitioned into twelve intervals, in each of which
the geometric properties of the possible leader-label intersections are similar
for all slopes. Depending on the particular slope interval, the amount $\iota(n,\delta,\theta)$ of
``interference'' between sub-problems  varies. This
significantly affects the algorithm's performance and leads to running times between $O(n^3\log n)$ and $O(n^3(\log n +
\iota(n,\delta,\theta)))=O(n^7)$
(Section~\ref{sec:Other_Directions}). Moreover, we can use our algorithm to
optimize the number of internal labels over all slopes $\theta$ at an increase
in running time by a factor of $O(n^2)$, as is shown in
Section~\ref{sub:optimize_directions}.

%

\paragraph{Problem Statement} We are given a map (or any other illustration)
\M, which we model for simplicity as a convex polygon (this is easy to relax to larger classes of well-behaved domains), and a set \P of $n$ points in \M that
must be labeled by rectangular labels (the bounding boxes of the label texts).
In addition, we are given a leader slope $\theta \in [0,2\pi)$. For simplicity
we assume that $\theta$ is none of the slopes defined by two points in \P. We
discuss in Section~\ref{sec:obst} how to remove this restriction. There
are two choices for assigning a label to a point $p \in \P$: either we assign an
\emph{internal label} $\lbl{p}$ on \M in a one-position model, or an
\emph{external label} outside of \M that is connected to $p$ with a leader
\ldr{p}. An internal label \lbl{p} is a rectangle that is anchored at $p$ by
its lower left corner. A leader \ldr{p} is a line segment of slope~$\theta$
inside \M; it may bend to the horizontal direction outside of \M in order to
connect to its horizontally aligned label, see Fig.~\ref{sfg:expl_nice}. So in
this model, the labeling is fixed once the choice for an internal or external
label has been made for each point $p \in \P$. For a \emph{valid} label assignment we
require that (i) the internal labels do not overlap each other or the leaders,
and that (ii) the leaders themselves do not intersect each other. Figure~\ref{fig:example} shows valid mixed labelings for three different slopes.

Given a set of points $P \subseteq \P$, let $\Lambda(P) =\{ \lbl{p} \mid p \in
P\}$ denote the set of (candidate) labels corresponding to the points in $P$ and let
$\Gamma(P)=\{ \ldr{p} \mid p \in P\}$ denote the set of (candidate) leaders corresponding
to the points in $P$.
A \emph{labeling} of \P is a partition of \P into sets \I and \E, the
points in \I labeled internally, the points in \E labeled externally, such that
no two labels in $\Lambda(\I)$ intersect, no two leaders in $\Gamma(\E)$
intersect, and no label from $\Lambda(\I)$ intersects a leader from
$\Gamma(\E)$.

For ease of presentation we first assume that all labels have the same size,
which, without loss of generality, we assume to be $1 \times 1$. Hence, an
internal label $\lbl{p}$ is a unit square with its bottom left corner on $p$.
This may be a realistic model in some settings (e.g., unit-size icons as
labels), but generally not all labels have the same size. We will sketch how to
relax this restriction in Section~\ref{sec:non-square}.

Each leader \ldr{p} can be split into an \emph{inner} part (or \emph{inner leader}), which is a line
segment of slope $\theta$ from $p$ to the intersection point with the
boundary of \M, and an \emph{outer} part (or \emph{outer leader}) from the boundary of \M
to the actual label. We focus our attention on the inner leaders as they determine how \P is separated into different subinstances. Hence we can basically think of the leaders as half-lines with slope $\theta$. 
For completeness, we explain a simple method of routing the outer leaders in Section~\ref{sec:place_external}.

It is well known that in general not all points in \P can be assigned an
internal label. The corresponding label number maximization problem is
\NP-hard~\cite{fw-ppwalm-91,ms-ccclp-91}, even if each label has just one candidate position~\cite{ln-sbolb-10}. If, however, all labels have the same position (e.g., to the top left of the anchor points) and no input point may be covered by any other label, the one-position case can be solved efficiently by first discarding all labels containing an anchor point and then applying a simple greedy algorithm on the resulting staircase patterns~\cite{ln-sbolb-10}. On the other hand, it is also known
that any instance can be labeled with external labels using efficient
algorithms~\cite{bksw-blmearm-07,bks-elcs-08}. 
Mixed labelings combine both label types and sit between the two extremes of purely internal and purely external labeling~\cite{bkps-ctlwbl-11,ln-sbolb-10}.
Here we are interested in
the \emph{internal label number maximization problem}, which was first studied
for $\theta \in \{0,\pi\}$ by Bekos et al.~\cite{bkps-ctlwbl-11}: Given a map \M, a set of points \P
in \M and a slope $\theta \in [0,2\pi)$, we wish to find a valid mixed labeling that
maximizes the number $|\I|$ of internally labeled points.

\section{Leaders from the Left}
\label{sec:left_leaders}

We start with the case that $\theta = 0$, i.e., all leaders are horizontal
half-lines leading from the points to the left of \M. Our approach for
maximizing the number of internal labels is to process the points in \P from
right to left and to recursively determine the optimal rightmost unprocessed
point $p$ to be assigned an external label. Since no leader may cross any
internal label, the leader \ldr{p} decomposes the current instance left of $p$
into two (almost) independent parts, one above \ldr{p} and one below. As it
turns out, a generic subinstance can be defined by an upper and a lower leader
shielding it from the outside and additional information about at most one
point outside the subinstance. The problem is then solved using dynamic
programming.

\subsection{Geometric Properties}
\label{sub:geometric_properties}

Let $p = (p_x,p_y)$ be a point in the plane and let $L_p = \{ q \mid q_x < p_x\}$ and $R_p =
\{ q \mid q_x > p_x \}$ denote the half-planes containing all points strictly
to the left and to the right of $p$, respectively. Analogously, we define the
half-planes $T_p$ and $B_p$ above and below $p$, respectively. Let
$\overline{S(\ell,u)} = T_\ell \cap B_u$ denote the horizontal slab defined by
points $\ell$ and $u$ (with $\ell_y < u_y$), and let $S(\ell,u) = \overline{S(\ell,u)} \cap L_\ell
\cap L_u$ denote the set of points in this slab that lie to the left of both $\ell$
and $u$, see Fig.~\ref{fig:subproblem}. We define $P_{\ell,u}$ as the subset of \P in $S(\ell,u)$ including $\ell$ and $u$, i.e., $P_{\ell,u} = \P \cap (S(\ell,u) \cup \{\ell,u\})$. With some abuse of notation we will sometimes also use $L_p$, $R_p$, $T_p$, and $B_p$ to mean the subset of \P that lies in the respective half-plane rather than the entire half-plane. 

\begin{figure}[tb]
  \centering
  \subfloat[]{\includegraphics[page=1]{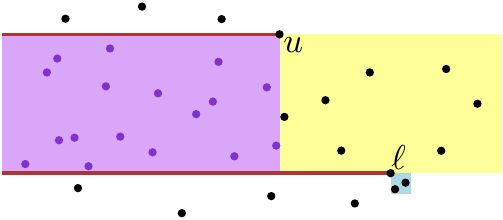}}
  \qquad
  \subfloat[\label{fig:subproblem-close}]{\includegraphics[page=5]{subproblem}}
    \caption{The slab $\overline{S(\ell,u)}$ in yellow, the region $S(\ell,u)$
    and its points from \P in purple, and the region $E(\ell)$ in blue.}
  \label{fig:subproblem}
\end{figure}


Recall that $\delta = \min\{n, 1/d_{\min}\}$ is a parameter that captures the maximum density of \P as the inverse of the smallest distance $d_{\min}$ between any two points in \P. We can use $\delta$ to bound the number of points in a unit square that may be labeled internally.
 
\begin{lemma}\label{lem:delta}
	At most $O(\delta)$ points in any unit square have a label $\lambda$ that does not contain another point in \P.
\end{lemma}

\begin{proof}
  Let $E$ be a unit square and let $E_{\P} = E \cap \P$ be the input points in
  $E$. Since the label \lbl{p} of every point $p \in E_\P$ is a unit square
  anchored by its lower left corner at $p$, no other point $q \in \P$ may lie
  to the top-right of $p$---otherwise $p$ must be labeled externally. Hence any
  set of points in $E_\P$ whose labels do not contain another point of $\P$
  must form a sequence $p^{(1)}, p^{(2)}, \dots, p^{(k)}$ such that $p^{(i)}_x
  < p^{(j)}_x$ and $p^{(i)}_y > p^{(j)}_y$ for any $i<j$. Since the minimum
  distance of any two points is $d_{\min}$ we immediately obtain that $E_\P$
  contains at most $O(\delta)$ points that can be labeled internally.
\end{proof}

Next, we characterize which leaders or labels outside of $\overline{S(\ell,u)}$ can interfere with a potential labeling of $P_{\ell,u}$ assuming $\ell$ and $u$ are labeled externally. 

\begin{lemma}
  \label{lem:no_point_above}
  Let $\ell,u \in \P$, let $(\I',\E')$, with $\ell,u \in \E'$ be a labeling of
  $P_{\ell,u}$. There is no point in $T_u \cup B_\ell$ whose leader intersects a label from $\Lambda(\I')$ and there is no point in $T_u$ whose label
  intersects a label from $\Lambda(\I')$.
\end{lemma}

\begin{proof}
  In any labeling of $P_{\ell,u}$ there are no intersections between labels and leaders. 
  In particular, no label \lbl{p} for $p \in \I'$ intersects \ldr{\ell} or \ldr{u}. 
  It follows that all labels for $\I'$ lie inside the slab $\overline{S(\ell,u)}$. 
  By definition, all leaders for $\E'$ also lie inside $\overline{S(\ell,u)}$.
  Since leaders of points in $T_u \cup B_\ell$ do not intersect $\overline{S(\ell,u)}$, no such leader can intersect a label from $\I'$.
  Moreover, all labels are anchored by their bottom left corner, and hence all labels of points in $T_u$ lie above $u$ and do not intersect $\overline{S(\ell,u)}$.
  Thus, no label for a point in $T_u$ can intersect a label in $\Lambda(\I')$. 
\end{proof}

It is not true, however, that labels for points in $B_\ell$ cannot intersect labels for $P_{\ell,u}$. Still, the influence of $B_\ell$ is very limited as the next lemma shows.
Let $E(\ell)$ denote the open unit square with top-left corner $\ell$, i.e.,
$E(\ell) = R_\ell \cap B_\ell \cap L_r \cap T_b$, where $r = (\ell_x + 1,
\ell_y)$ and $b = (\ell_x,\ell_y - 1)$. See
Fig.~\ref{fig:subproblem-close}.

\begin{lemma}
  \label{lem:one_point_below}
  Let $\ell,u \in \P$, let $(\I',\E')$, with $\ell, u \in \E'$ be a labeling of
  $P_{\ell,u}$, and let $(\I'',\E'')$ denote a
  labeling of $\P \cap B_\ell \cup \{\ell\}$ with $\ell \in \E''$. There is at most one point $p \in \I''$ whose
  label may intersect a label of $\I'$, 
  and $p \in E(\ell)$. 
\end{lemma}

\begin{proof}
  Since $\ell \in \E' \cap \E''$ we know that no label in $\Lambda(\I')$ or
  $\Lambda(\I'')$ intersects $\ldr{\ell}$.  Thus $\ldr{\ell}$ serves as a
  separation line between the labels $\Lambda(\I')$ and $\Lambda(\I'')$.  Let
  $\P'' \subseteq \I''$ denote the set of points whose labels intersect a label
  of $\I'$, and let $\Lambda'' = \Lambda(\P'')$ denote the corresponding set of
  labels.
  We first argue that $\P'' \subseteq E(\ell)$. Then we argue that $\P''$ can
  contain at most one point. 

  The labels in $\Lambda''$ do not intersect $\ldr{\ell}$, hence they lie strictly
  right of $\ell$. Thus, $\P'' \subseteq R_\ell$.
  All points in $\I'$ lie to the left of $\ell$, and their labels have width
  one. All labels in $\Lambda''$ intersect such a label, thus all points in
  $\P''$ must lie in $L_r$, where $r=(\ell_x+1,\ell_y)$.
  All labels in $\Lambda''$ intersect a label from a point in $\I'$. Thus, all
  labels in $\Lambda''$ intersect the horizontal line containing
  $\ldr{\ell}$. Since all labels have height one, it follows that all points in
  $\P''$ lie in $T_b$, where $b=(\ell_x,\ell_y - 1)$.
  By definition the points in $\P''$ lie in $B_\ell$. So we have $\P''
  \subseteq R_\ell \cap L_r \cap T_b \cap B_\ell = E(\ell)$.

  The labels in $\Lambda''$ are pairwise disjoint and all intersect the top
  side $s$ of $E(\ell)$. Since the length of $s$ is smaller than one, each
  label in $\Lambda''$ has width exactly one, and all labels lie in $R_\ell$ it
  follows that there can be at most one label in $\Lambda''$. Thus, there is
  also at most one point in $\P''$.
\end{proof}

From Lemma~\ref{lem:no_point_above} and Lemma~\ref{lem:one_point_below} it follows that if $\ell$ and $u$ are
labeled externally, there is at most one point $r$ below $\ell$ that can
influence the labeling of the points in $S(\ell,u)$.

\subsection{Computing an Optimal Labeling}
\label{sub:Computing_a_Labeling_left}

We define \lab{\ell,u,r}, with $\ell, u \in \P$, and $r \in E(\ell) \cup
\{\bot\}$ as the maximum number of points in $S(\ell,u)$ that can be labeled
internally, given that

\begin{enumerate}[noitemsep,label=\emph{(\roman*)}]
\item the points $\ell$ and $u$ are labeled externally,
\item all remaining points in $\overline{S(\ell,u)} \setminus S(\ell,u)$ have been labeled
  internally, and
\item point $r$ is labeled internally. If $r = \bot$ then no point in $E(\ell)$
  is labeled internally.
\end{enumerate}

See Fig.~\ref{fig:decomposition_subproblem}(a) for an
illustration. Furthermore, given $\ell$, $r$, and a point $p \in T_\ell$
we define $\varrho(p,\ell,r)$ to be the topmost point in $E(p) \cap
(T_\ell \cup \{r\})$ if such a point exists. Otherwise we define
$\varrho(p,\ell,r) = \bot$. 
%


\begin{figure}[t]
  \centering
  \subfloat[]{\includegraphics[page=2]{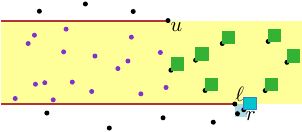}}
  \qquad
  \subfloat[\label{sfg:decompose}]{\includegraphics[page=3]{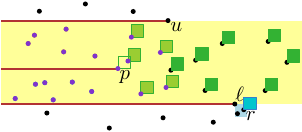}}
  \caption{(a) \lab{\ell,u,r} expresses the maximum number of points in
    $S(\ell,u)$ (purple), that can be labeled internally in the depicted
    situation. (b) The rightmost point $p$ that is labeled with an external
    label decomposes the problem into two subproblems (the orange and blue
    points).}
  \label{fig:decomposition_subproblem}
\end{figure}

\begin{lemma}
  \label{lem:combining}
  For any $\ell, u \in \P$, and $r \in E(\ell) \cup \{\bot\}$, we have that
  $\lab{\ell,u,r} = |S(\ell,u)|$, or
  $\lab{\ell,u,r} = |R_p \cap S(\ell,u)| + \lab{\ell,p,r} + \lab{p,u,r'}$,
  where $p$ is the rightmost point in $S(\ell,u)$ with an external label and $r' = \varrho(p,\ell,r)$.
%
\end{lemma}

\begin{proof}
  Let $(\I^*,\E^*)$ be an optimal labeling of $S(\ell,u)$ that satisfies the
  constraints \emph{(i)--(iii)} on $\lab{\ell,u,r}$, i.e.,  $\lab{\ell,u,r} = |\I^*|$. In case $\E^* = \emptyset$, we have $\lab{\ell,u,r} = |S(\ell,u)|$ and the lemma trivially
  holds. Otherwise, there must be a rightmost point $p \in \E^*$ with an
  external label. Consider the partition of $\I^*$ at point $p$ into the lower left part $B^* = B_p \cap L_p
  \cap \I^*$, the upper left part $T^* = T_p \cap L_p \cap \I^*$, and the right part $R^* = R_p \cap \I^*$, see
  Fig.~\ref{sfg:decompose}. We show that $|R^*| = |R_p \cap
  S(\ell,u)|$, $|B^*| = \lab{\ell,p,r}$, and $|T^*| =
  \lab{p,u,r'}$, which proves the lemma.

  Since $p$ is the rightmost point with an external label it follows that all
  points in $S(\ell,u)$ right of $p$ are labeled internally. Hence, $R^* = R_p
  \cap S(\ell,u)$.

  Next, we observe that $\LL_B = (B^*,S(\ell,p) \setminus B^*)$ as a sub-labeling of $(\I^*,\E^*)$ forms a valid labeling
  of $S(\ell,p)$, so by Lemma~\ref{lem:one_point_below} there is at most one
  point $\hat{r}$ below $\ell$ that can influence the labeling of
  $S(\ell,p)$. This point $\hat{r}$, if it exists, lies in $E(\ell)$. By constraint \emph{(iii)} point $r$
  lies in $E(\ell)$ or $r = \bot$ and no point in $E(\ell)$ is labeled internally, and thus $r$ can be the only point in $E(\ell)$ labeled
  internally, i.e., $\hat{r} = r$. So, we have that \emph{(i)} $\ell$ and $p$
  are labeled externally, \emph{(ii)} all points in $\overline{S(\ell,p)}
  \setminus S(\ell,p)$ are labeled internally, and \emph{(iii)} point $r$ is
  the only internally labeled point in $E(\ell)$. Thus the definition of $\Phi$ applies and we obtain $|B^*| \leq \Phi(\ell,p,r)$.

  Lemmas~\ref{lem:no_point_above} and \ref{lem:one_point_below} together imply
  that any labeling of $S(\ell,p)$ is independent from any labeling of
  $S(p,u)$. Thus, it follows that $\LL_B$ is an optimal labeling of $S(\ell,p)$
  (given the constraints), since otherwise $(\I^*,\E^*)$ could also be improved.
  Thus $|B^*| \geq \lab{\ell,p,r}$ and we obtain $|B^*|
  = \lab{\ell,p,r}$.

  Finally, we consider the upper left part $\T^*$.
  By Lemma~\ref{lem:one_point_below} there is at most one point $r'$ in $B_p$
  with an internal label that can influence the labeling of $S(p,u)$ and we have $r' \in E(p)$. 
  We need to show that $r' = \varrho(p,\ell,r)$. 
  Then the rest of the argument is analogous to the argument for $B^*$.

  We claim that $r'$ is the topmost point in $E(p) \cap (T_\ell \cup \{r\})$. 
  Assume that $r' \not\in T_\ell$, which means $r' \in B_\ell$. 
  We know that $\ldr{\ell}$ does not intersect $\lbl{r'}$ and hence $r' \in R_\ell$.
  This means that $r' \in E(p) \cap B_\ell \cap R_\ell =: X$ and since $p$ lies to the top-left of $\ell$ we have $X \subseteq E(\ell)$. 
  By definition $r$ is the only point with an internal label in $E(\ell)$ and hence $r' = r$. 
  So if $r' \ne r$ we have $r' \in E(p) \cap T_\ell$.
  Now assume that $r' \ne \varrho(p,\ell,r)$. 
  Then there is another point $q \in E(p) \cap T_\ell$ above $r'$. 
  This point $q$ must be labeled externally since no two points in $E(p)$ can be labeled internally.
  This is a contradiction since by definition $p$ is the rightmost externally labeled point in $S(\ell,u)$ and by constraint \emph{(ii)} all points in $\overline{S(\ell,u)} \setminus S(\ell,u)$ are labeled internally.
  So indeed $r' = \varrho(p,\ell,r)$ and the same arguments as for $B^*$ can be used to obtain $|T^*| = \lab{p,u,r'}$.
  %
\end{proof}

Let $\ell$, $u \in \P$, and $p \in S(\ell,u)$. We observe that $|S(\ell,p)|$
and $|S(p,u)|$ are strictly smaller than $|S(\ell,u)|$. Thus,
Lemma~\ref{lem:combining} gives us a proper recursive definition for~$\Phi$:
\begin{align*}
  \lab{\ell,u,r} = \max \big\{&\Psi(S(\ell,u)), \\
    & \max_{p \in S(\ell,u)}  \left\{\Psi(R_p \cap S(\ell,u))  + \lab{\ell,p,r} + \lab{p,u,\varrho(p,\ell,r)}\right\} \big\},
\end{align*}
%
\noindent where
\[
\Psi(P) =
\begin{cases}
  |P| & \parbox{0.7\textwidth}{if all labels in $\Lambda(P \cup \{r\} \cup
    (\overline{S(\ell,u)} \setminus S(\ell,u)))$ are pairwise disjoint, and
    their intersection with \ldr{\ell} and \ldr{u} is empty,} \\
  -\infty & \text{otherwise}.
\end{cases}
\]
 
We can now express the maximum number of points in \P that can be labeled
internally using $\Phi$. We add two dummy points to \P that we assume are
labeled externally: a point $p_\infty$ that lies sufficiently far above and
to the right of all points in \P, and a point $p_{-\infty}$ below and to the
right of all points in \P. The maximum number of points labeled internally is
then \lab{p_{-\infty},p_\infty,\bot}.

We use dynamic programming to compute
$\lab{\ell,u,r}$ for all $\ell,u \in \P \cup \{p_\infty, p_{-\infty}\}$ with $\ell_y < u_y$ and $r \in E(\ell) \cup \{\bot\}$. By finding the maximum in a set of linear size, each value $\lab{\ell,u,r}$ can be computed in $O(n)$ time, given
that the values \lab{\ell',u',r'} for all subproblems have already been
computed and stored in a table and the relevant values for the functions $\varrho$ and $f$ have been precomputed. There are $O(n)$ choices for each of $\ell$ and $u$; further there are $O(\delta)$ choices for the point $r$ given $\ell$ since $r$ is labeled internally and we know from Lemma~\ref{lem:delta} that there are at most $O(\delta)$ points in $E(\ell)$ as candidates for an internal label. This
results in an $O(n^3 \delta)$ time and $O(n^2 \delta)$ space
dynamic-programming algorithm. We show next that the preprocessing of $\varrho$ and $f$ can be done in $O(n^3 \log n)$ time. 

To compute $\lab{\ell,u,r}$ we actually have to compute $\varrho(p,\ell,r)$ and
$\Psi'(p,r) := \Psi(R_p \cap S(\ell,u))$ for all points $p \in S(\ell,u)$.
We can preprocess all points in \P in $O(n \log n)$ time, such that we can compute each
$\varrho(p,\ell,r)$ in $O(1)$ time as follows. First, we compute and store for each
point $p \in \P$ the topmost point $q_p \in \P$ in $E(p)$. This requires $n$ standard
priority range queries that take $O(n \log n)$ time in total using priority range trees with fractional cascading~\cite[Chapter 5]{bkos-cgaa-00}. 
To compute $\varrho(p,\ell,r)$ we
then check if $q_p$ lies above $\ell$. If it does, we have $\varrho(p,\ell,r) =
q_p$. Otherwise, the only candidate point for $\varrho(p,\ell,r)$ is $r$ and we
can check in $O(1)$ time if $r$ lies in $E(p)$.
This takes $O(1)$ time for each triple $(p,\ell,r)$ and $O(n^2 \delta)$ time in total. 

Next, we fix $\ell$ and $u$, and compute a representation of $\Psi'$ in $O(n
\log n)$ time, such that for each  $p \in
S(\ell,u)$ and $r \in E(\ell) \cup \{\bot \}$ we can obtain $\Psi'(p,r)$ in constant time.

We start by computing the values $\Psi'(p,\bot)$, for all $p$. We sweep a vertical
line from right to left. That is, we sort all points in $\overline{S(\ell,u)}$
by decreasing $x$-coordinate, and process the points in that order. The status
structure of the sweep line contains the number of points $N$ in $S(\ell,u)$
right of the sweep line, and a (semi-)dynamic data structure \T, which stores
the labels from the points right of the sweep line, and can report all labels
intersected by an (axis-parallel) rectangular query window. All labels are unit
squares, so \lbl{r} intersects a label \lbl{q} if and only if \lbl{r} contains
a corner point of \lbl{q}. Furthermore, we only ever insert new labels (points)
into \T, thus it suffices if \T supports only insert and query
operations. It follows that we can implement \T using a semi-dynamic range tree
using dynamic fractional cascading~\cite{mn-dfc-90}. In this data structure insertions and queries take $O(\log n)$ time.

When we encounter a new point $p$, $p \not\in \{\ell, u\}$ we test if the label
of $p$ intersects any of the labels encountered so far. We can test this using
a range query in the tree \T. If $p \in S(\ell,u)$ we also explicitly test if
\lbl{p} intersects $\ldr{\ell}$ or $\ldr{u}$. If there are no points in the
query range \lbl{p}, and \lbl{p} does not intersect \ldr{\ell} or \ldr{u} we
report $\Psi'(p,\bot) = N$, increment $N$ (if applicable), and insert the corner
points of \lbl{p} into \T. If the query range \lbl{p} is not empty, it follows
that $\Psi'(p',\bot) = -\infty$, for $p'=p$ as well as for any point to the left
of $p$. Hence, we report that and stop the sweep. Our algorithm runs in $O(n
\log n)$ time: sorting all points takes $O(n \log n)$ time, and handling each
of the $O(n)$ events takes $O(\log n)$ time.



Now consider a point $r \in E(\ell)$. We observe that for all points $p$ right
of $r$, we have that $\Psi'(p,r) = \Psi'(p,\bot) = 0$ since $r$ is right of all
points in $S(\ell,u)$. Consider the points left of $r$ ordered by decreasing
$x$-coordinate. There are two options, depending on whether or not \lbl{r}
intersects the label \lbl{p} of the current point $p$. If \lbl{r} intersects
\lbl{p}, we have $\Psi'(p,r) = -\infty$ as well as $\Psi'(p',r) = -\infty$ for all
points $p'$ left of $p$. If \lbl{r} does not intersect \lbl{p} we still have
$\Psi'(p,r) = \Psi'(p,\bot)$. We can test if \lbl{r} intersects any other label using
a range priority query with \lbl{r} in (the final version of) the range tree
\T. We need $O(\delta)$ such queries, which take $O(\log n)$ time each. This gives a total running time of $O(n \log n)$. We then conclude:


The above algorithm can
also be used when the leaders have a slope $\theta \neq 0$. However, the data
structure \T that we use is fairly complicated. In this specific case where
$\theta = 0$, we can also use a much easier data structure, and still get a
total running time of $O(n \log n)$. Instead of using the semi-dynamic range
tree as status structure, we use a simple balanced binary search tree that
stores (the end-points of) a set of vertical \emph{forbidden intervals}. When
we encounter a new point $p$, we check if $p_y$ lies in a forbidden
interval. If this is the case then \lbl{p} intersects another label. Otherwise
we can label $p$ internally. This set of forbidden intervals is easily
maintained in $O(\log n)$ time.


We use this algorithm for every pair $(\ell,u)$. Hence, after a total of $O(n^3
\log n)$  preprocessing time, we can answer $\Psi'(p,r)$ queries for any $p$ and
$r$ in constant time. This yields the following result, which improves the
previously best known pseudo-polynomial $O(n^{\log n +3})$-time algorithm of
Bekos et al.~\cite{bkps-ctlwbl-11} for the left-sided case $\theta = 0$.

\begin{theorem}
  \label{thm:left_leaders}
  Given a set \P of $n$ points, we can compute a labeling of \P that maximizes
  the number of internal labels for $\theta = 0$ in $O(n^3\log n +
  n^3 \delta)$ time and $O(n^2 \delta)$ space, where $\delta = \min \{n, 1/d_{\min}\}$ for the minimum distance $d_{\min}$ in \P.
\end{theorem}


\section{Other Leader Directions}
\label{sec:Other_Directions}

For other leader slopes $\theta \ne 0$ we use a similar approach as before. We
consider a sub-problem $S(\ell,u)$ defined by two externally labeled points
$\ell$ and $u$. We again find the ``rightmost'' point in the slab labeled
externally. This gives us two sub-problems, which we solve recursively using
dynamic programming. However, there are four complications:

\begin{itemize}[noitemsep]
\item The region $E(\ell)$ containing the points ``below'' the slab
  $\overline{S(\ell,u)}$ that can influence the labeling of $S(\ell,u)$ is no
  longer a unit square. Depending on the orientation, it can contain more
  than one point with an internal label. See Fig.~\ref{fig:regions_E_and_F}(a).
\item In addition to the region $E(\ell)$, which contains points that can
  interfere with a subproblem from below, we now also need to consider a second
  region, which we call $F(u)$, containing points whose labels can interfere
  with a subproblem from above. See Fig.~\ref{fig:regions_E_and_F}(b).
\item The labels of points in $S(\ell,u)$ are no longer fully contained in the slab
  $\overline{S(\ell,u)}$. Hence, we have to check that they do not intersect
  with leaders of points outside $\overline{S(\ell,u)}$. See Fig.~\ref{fig:regions_E_and_F}(b).
\item The regions $E(p)$ and $F(p)$ are no longer strictly to ``the right'' of
  $p$. Hence, for some sub-problems we may have already decided (by definition
  of the sub-problem) that a point $q \in E(p)$ that lies ``left'' of $p$ is
  labeled internally. Hence, we can no longer choose $q$ to be the rightmost
  point in $S(\ell,p)$ to be labeled externally.
\end{itemize}


\begin{figure}[tb]
  \centering
  \includegraphics{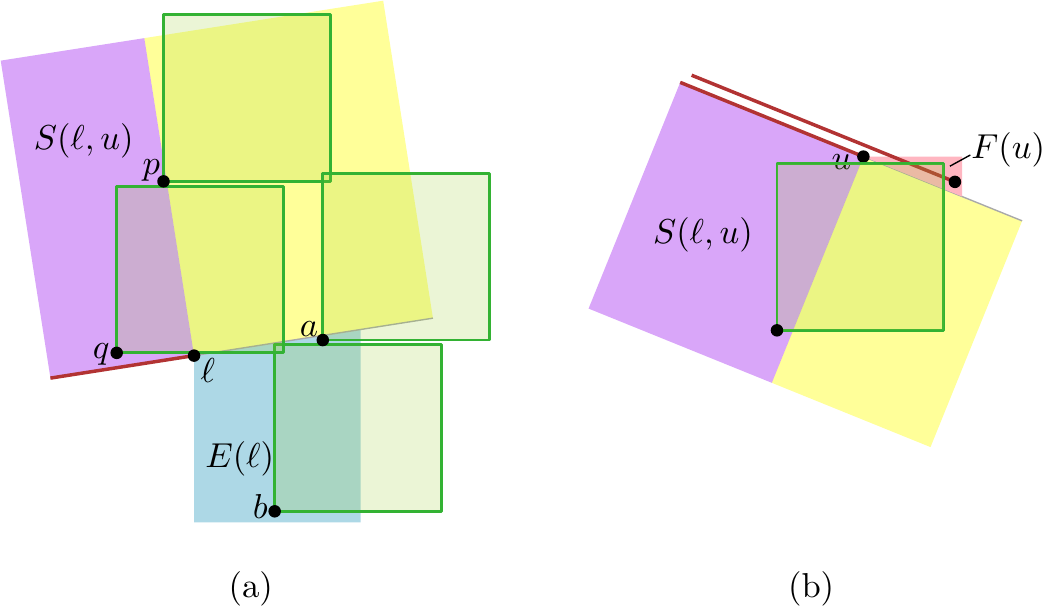}
  \caption{(a) There may be more than one point ``below'' $\ell$ with an
    internal label if the leaders arrive from the bottom left. (b) For other
    directions there is also a region $F(u)$ ``above'' the subproblem that can
    influence the labeling of $S(\ell,u)$.}
  \label{fig:regions_E_and_F}
\end{figure}

We start by explicitly finding the points in \P whose internal labels contain
other points. We are forced to label these points externally. It is easy to
find those points in $O(n^2)$ time in total. Let $\P_X$ denote this set of
points. Additionally, we spend $O(n^2)$ time to mark each point if its label
intersects $\Gamma(\P_X)$. Hence, for each point we can determine in constant
time if it intersects $\Gamma(\P_X)$. For ease of notation we will write \P to
mean the set $\P \setminus \P_X$ in the remainder of this section.
 \maarten {We
  don't need to know their leaders if we also discard any points whose internal
  labels intersect these forced leaders. If we recursively test this, we can
  really remove the points. Not sure if we want to do that, though, since it
  appears unnecessary---it might be conceptually easier to think about.}

Let $\theta$ be the given orientation for the leaders. Consider conceptually
rotating the coordinate system such that orientation $\theta$ corresponds to
the negative $x$-axis as in the previous section, and let $\tilde{B_p},
\tilde{T_p}, \tilde{L_p},$ and $\tilde{R_p}$ denote the points in \P in the
bottom, top, left, and right half-planes bounded by $p$ with respect to this
coordinate system. Analogously, we define $\overline{S(\ell,u)} =
\tilde{T_\ell} \cap \tilde{B_u}$, and $S(\ell,u) = \overline{S(\ell,u)} \cap
\tilde{L_\ell} \cap \tilde{L_u}$.

Our goal is  again to bound the number of different labelings of the points
in $\tilde{B_\ell}$ and $\tilde{T_u}$ that can influence the labeling of
$S(\ell,u)$. We will show that whether a labeling of $\tilde{B_\ell}$
influences the labeling of $S(\ell,u)$ depends only on a small subset of the
points in $\tilde{B_\ell}$. We refer to a labeling of $\tilde{B_\ell}$
restricted to those points as a \emph{configuration} of $\tilde{B_\ell}$. The
same holds for the points in $\tilde{T_u}$.

\subsection{Bounding the number of Configurations}
\label{sub:Bounding_the_number_of_Configurations}

We start by bounding the number of configurations of $\tilde{B_\ell}$. Let
$E(\ell)$ denote the \emph{bottom influence region} of $\ell$. That is, the
points in $\tilde{B_\ell}$ ``below'' the slab $\overline{S(\ell,u)}$ whose
label can intersect a label of a point in $S(\ell,u)$. Fig.~\ref{fig:region_E}
shows the regions $E(\ell)$ for various orientations of the leaders.

Similarly, we can define a region $E'(\ell) \subset \tilde{B_\ell}$ such that
the \emph{leaders} of points in $E'(\ell)$ can intersect a label of a point in
$S(\ell,u)$.

\begin{lemma}
  \label{lem:sizes_E}
  For a subproblem $S(\ell,u)$ the size of the bottom influence region
  $E(\ell)$ is at most $1 \times e(\theta)$ or $e(\theta) \times 1$, where

  \hspace{-1cm}
  \begin{tabular}{>{\centering\arraybackslash}m{0.34\textwidth}
                  >{\centering\arraybackslash}m{0.35\textwidth}
                  >{\centering\arraybackslash}m{0.3\textwidth}}
  $\displaystyle
  e(\theta) \leq
  \begin{cases}
    1 & \text{if } \theta = 0 \\
    2 & \text{if } \theta \in (0,\pi/4) \\
    1 & \text{if } \theta \in [\pi/4,\pi/2) \\
    0 & \text{if } \theta = \pi/2  \\
    1 & \text{if } \theta \in (\pi/2,3\pi/4) \\
  \end{cases}
  $
  &
  $\displaystyle
  e(\theta) \leq
  \begin{cases}
    0 & \text{if } \theta \in [3\pi/4,5\pi/4] \\
    1 & \text{if } \theta \in (5\pi/4,3\pi/2) \\
    0 & \text{if } \theta  = 3\pi/2 \\
    3 & \text{if } \theta \in (3\pi/2,7\pi/4) \\
    2 & \text{if } \theta \in [7\pi/4,2\pi). \\
  \end{cases}
  $
  &
  \includegraphics[page=1]{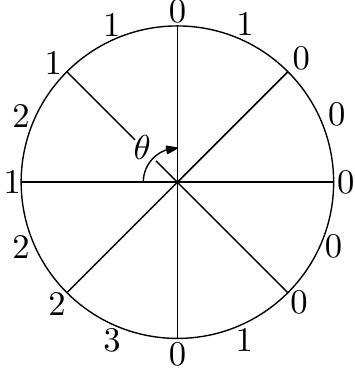}
  \end{tabular}
\end{lemma}


\begin{figure}[tbp]
  \centering
  \includegraphics[width=0.24\textwidth,page=1]{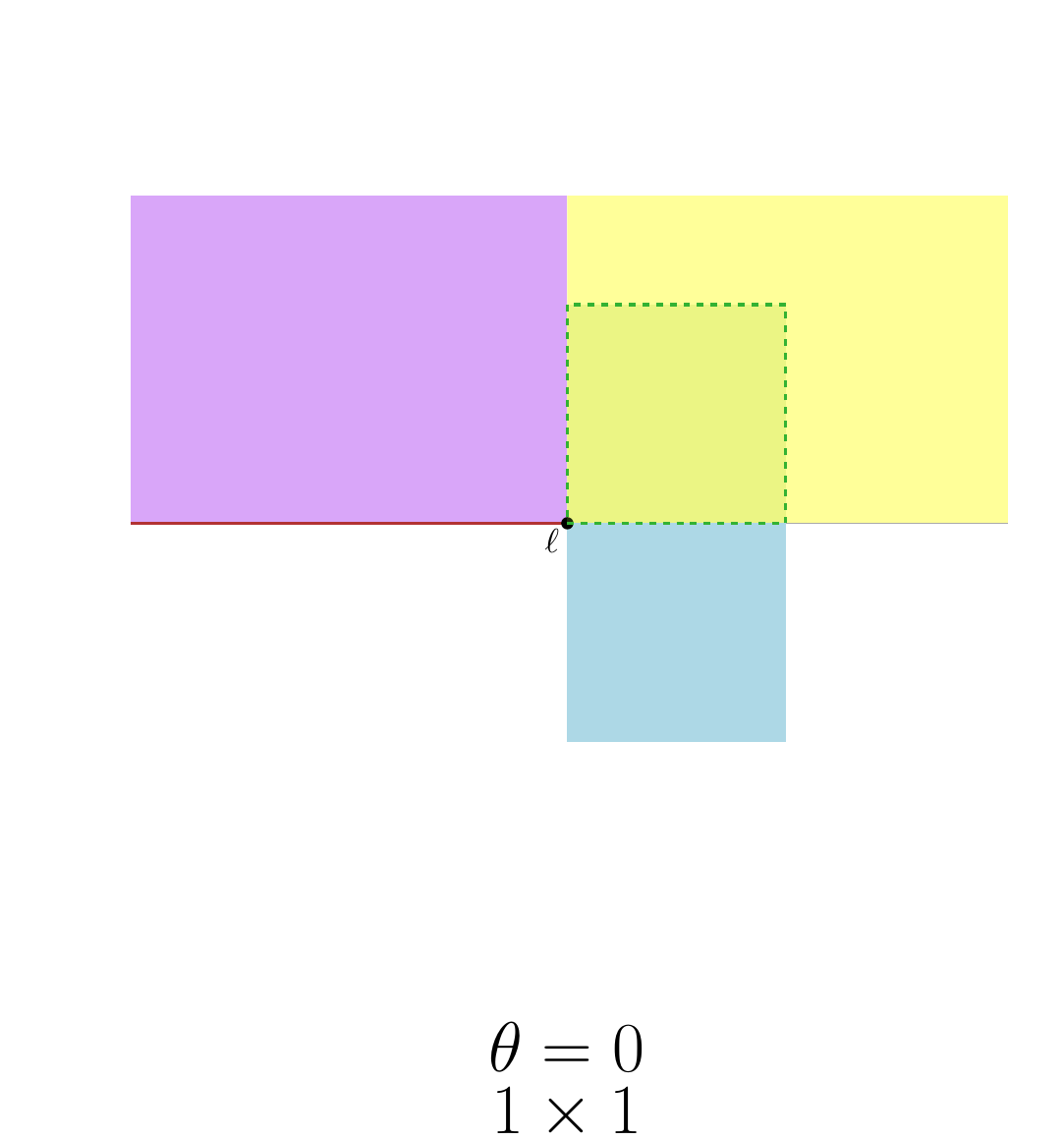}
  \includegraphics[width=0.24\textwidth,page=2]{bottom_region_E}
  \includegraphics[width=0.24\textwidth,page=3]{bottom_region_E}
  \includegraphics[width=0.24\textwidth,page=4]{bottom_region_E}
  \\
  \includegraphics[width=0.24\textwidth,page=5]{bottom_region_E}
  \includegraphics[width=0.24\textwidth,page=6]{bottom_region_E}
  \includegraphics[width=0.24\textwidth,page=7]{bottom_region_E}
  \includegraphics[width=0.24\textwidth,page=8]{bottom_region_E}
  \\
  \includegraphics[width=0.24\textwidth,page=9]{bottom_region_E}
  \includegraphics[width=0.24\textwidth,page=10]{bottom_region_E}
  \includegraphics[width=0.24\textwidth,page=11]{bottom_region_E}
  \includegraphics[width=0.24\textwidth,page=12]{bottom_region_E}
  \caption{The region $E(\ell)$ (blue) with respect to sub-problem $S(\ell,u)$
    (purple), depending on the orientation $\theta$ of the leaders. For each
    orientation we give an upper bound on the size of $E(\ell)$.}
  \label{fig:region_E}
\end{figure}

\begin{proof}
  We prove this by case distinction on $\theta$.

  \begin{description}[leftmargin=1em,style=sameline,font=\textbf,noitemsep,nosep]
    \item[case $\theta = 0$.] See Lemma~\ref{lem:one_point_below}.
    \item[case $\theta \in (0,\pi/4)$.]
      \maarten [suggests] {Perhaps refer directly to Figure 6(b)? And the same for each case?}
       All labels in $\Lambda(S(\ell,u))$ lie in
      $T_\ell$ and in $\tilde{L_r}$, where $r = (\ell_x+1,\ell_y)$. It follows that
      $E(\ell) \subseteq T_{\ell'} \cap \tilde{L_{r'}}$, where $\ell' =
      (\ell_x,\ell_y-1)$ and $r'=(r_x,r_y-1)$. Furthermore, it is easy to see
      that $E(\ell) \subset R_\ell$, and that $E(\ell) \subset
      \tilde{B_\ell}$. Hence, $E(\ell) \subset T_{\ell'} \cap R_\ell \cap
      \tilde{L_{r'}} \cap \tilde{B_\ell}$.

      Since the leaders are sloped downwards it follows that the height of
      $E(\ell)$ is at most one. The maximum width of $E(\ell)$ is realized by
      $\ell$ and the intersection point $p$ of the lines bounding
      $\tilde{B_\ell}$ and $\tilde{L_{r'}}$. Using that $\theta \in (0,\pi/4)$
      basic trigonometry shows that the width is at most two.
    \item[case $\theta \in [\pi/4,\pi/2)$.] Similar to the previous
      case. However, now the width is determined by the intersection between
      $T_{\ell'}$ and $\tilde{B_\ell}$. From basic trigonometry it then follows
      that the width is at most one.
    \item[case $\theta = \pi/2$.] For labels from $E(\ell)$ to intersect
      $\Lambda(S(\ell,u))$ we need $E(\ell) \subset T_{\ell'}$, where
      $\ell'=(\ell_x,\ell_y-1)$. However, to avoid intersecting \ldr{\ell} we
      need $E(\ell) \subset B_{\ell'}$. It follows that $E(\ell) = \emptyset$.
    \item[case $\theta \in (\pi/2,3\pi/4)$.] Using similar arguments as before
      it follows that $E(\ell) \subset B_{\ell'} \cap \tilde{L_{\ell'}} \cap
      \tilde{B_\ell} \cap L_\ell \cap R_{\ell'}$, where $\ell' =
      (\ell_x-1,\ell_y-1)$. Since $E(\ell) \subset L_\ell \cap R_{\ell'}$ the
      width is at most one. Basic trigonometry again shows that the height is
      at most one.
    \item[case $\theta \in {[3\pi/4,5\pi/4]}$.] For the sub-case $\theta \in
      [3\pi/4,\pi)$ the lines bounding $\tilde{L_{\ell}}$ and $R_{\ell'}$, with
      $\ell'=(\ell_x-1,\ell_y-1)$ intersect below the line containing
      \ldr{\ell}. We then obtain $E(\ell) \subset R_{\ell'} \cap B_{\ell'} \cap
      \tilde{B_\ell} = \emptyset$. In the remaining sub-case $\theta \in
      [\pi,5\pi/4]$ the regions $\Lambda(S(\ell,u))$ and
      $\Lambda(\tilde{B_\ell})$ are disjoint. It follows that $E(\ell)$ is empty.
    \item[case $\theta \in {(5\pi/4,3\pi/2)}$.] Point $\ell$ is now the point
      with the maximum $y$-coordinate. It then follows that all labels of
      $S(\ell,u)$ lie in $B_{\ell'}$, where $\ell'=(\ell_x,\ell_y+1)$. Hence,
      we also get $E(\ell) \subset B_{\ell'}$. The labels of $S(\ell,u)$ do not
      intersect \ldr{\ell}, hence they are contained in $L_\ell \cup
      \tilde{T_\ell}$. We then have $E(\ell) \subset \tilde{B_\ell} \cap
      (L_\ell \cup \tilde{T_\ell}) = \tilde{B_\ell} \cap L_\ell$. Since $\theta \in
      (3\pi/4,3\pi/2)$ it now follows that the height and width are both at
      most one.
    \item[case $\theta = 3\pi/2$.] All labels from $S(\ell,u)$ lie in $L_\ell$,
      all labels from $\tilde{B_\ell} = R_\ell$ lie in $R_\ell$. Hence,
      $E(\ell) = \emptyset$.
    \item[case $\theta \in (3\pi/2,7\pi/4)$.] It is again easy to show that
      $E(\ell) \subset R_\ell \cap L_r$, where $r=(\ell_x+1,\ell_y+1)$, and
      thus has width (at most) one. All labels from points in $S(\ell,u)$ have
      width and height one, and are thus contained in
      $\tilde{L_r}$. Furthermore, they do not intersect \ldr{\ell}, from which
      we obtain that they are contained in $\tilde{T_\ell} \cup T_\ell$. From
      the former we get that $E(\ell) \subset
      \tilde{L_r}$. From the latter
      we get that $E(\ell) \subset \tilde{T_{\ell'}} \cup T_{\ell'}$, where
      $\ell'=(\ell_x,\ell_y-1)$. Hence, we obtain $E(\ell) \subset R_\ell \cap
      L_r \cap \tilde{L_r} \cap (\tilde{T_{\ell'}} \cup T_{\ell'}) \cap
      \tilde{B_\ell}$.

      Since $\theta \in (3\pi/2,7\pi/4)$ the height of $E(\ell)$ is determined
      by $\ell'$ and the intersection point $p$ between $\tilde{B_\ell}$ and
      $\tilde{L_r}$. Trigonometry now shows that the height is at most three.
    \item[case $\theta \in [7\pi/4,2\pi)$.] Similar to the previous case we get
      a width of at most one. The height is now determined by $\ell'$ and the
      intersection of $\tilde{B_\ell}$ and $R_r$. Since $\theta \in
      [7\pi/4,2\pi)$ the height is at most two.
  \end{description}
  \vspace{-\baselineskip}
\end{proof}

\begin{corollary}
  \label{cor:at_most_KE_labeled_internally}
  There can be at most $e(\theta)$ points in $E(\ell)$ labeled internally
  such that their labels are disjoint.
\end{corollary}

Next, we turn our attention to the points in $E'(\ell)$ whose leader can
intersect a label of a point in $S(\ell,u)$. A leader of a point in
$\tilde{B_\ell}$, and thus in $E'(\ell)$, can intersect a label of $S(\ell,u)$
only if the labels intersect $\tilde{B_\ell}$. This happens only if $\theta \in
(5\pi/4,3\pi/2)$ or $\theta \in (3\pi/2,2\pi)$. See Fig.~\ref{fig:region_E}. In
the former case we thus have $E'(\ell) = \tilde{B_\ell} \cap \tilde{T_{\ell'}}
\cap L_\ell$, where $\ell'=(\ell_x,\ell_y + 1- \tan
(\theta-5\pi/4))$, and in the latter case we have $E'(\ell) = \tilde{B_\ell}
\cap \tilde{T_z} \cap T_z$, where $z=(\ell_x+1,\ell_y)$. See
Fig.~\ref{fig:leaders_in_bottom}.

We now note that if we label the set $Q \subseteq E(\ell)$ internally, then all
other points in $E(\ell)$ are labeled externally. Hence, if the leaders of the
remaining points (e.g. those in $E(\ell)\setminus Q$) intersect with labels of
points in $S(\ell,u)$, this is already captured by the configuration involving
$Q$. Therefore, the points in $E(\ell)\setminus Q$ themselves do not define new
configurations. Similarly, the points in $\P_X$ do not define any new
configurations.

The points in $E'(\ell)$ that lie outside of $E(\ell)$ can still be labeled
both internally or externally. Let $E''(\ell) = E'(\ell) \setminus E(\ell)$
denote the region containing these points. We now observe:

\begin{figure}[t]
  \centering
  \includegraphics{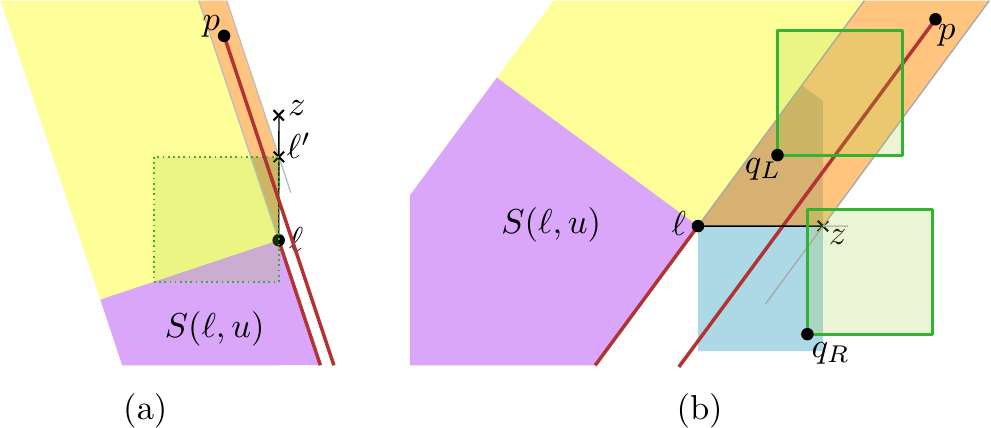}
  \caption{The region $E'(\ell)$ (in pink) for the cases $\theta \in
    (5\pi/4,3\pi/2)$ (a) and $\theta \in (3\pi/2,2\pi)$ (b). In both cases the
    leader \ldr{p} of a point $p \in E''(\ell)$ intersects the line segment
    $\overline{\ell{}z}$ and thus subdivides $E(\ell)$ into a left region $L$
    and a right region $R$.}
  \label{fig:leaders_in_bottom}
\end{figure}

\begin{lemma}
  \label{lem:leaders_in_E''}
  Let $Q$ be the set of points in $E(\ell)$ labeled internally, and let $p \in
  E''(\ell)$ be labeled externally. For all points $q \in Q$ we have: $q$
  intersects the leader \ldr{p}, or if there is a label \lbl{a}, $a \in
  S(\ell,u)$, that intersects \lbl{q}, then it also intersects the leader
  \ldr{p}.
\end{lemma}

\begin{proof}
  It is easy to see that any leader \ldr{p} intersects the line segment
  $\overline{\ell{}z}$, with $z=(\ell_x,\ell_y+1)$ if $\theta \in
  (5\pi/4,3\pi/2)$, and $z=(\ell_x+1,\ell_y)$ if $\theta \in
  (3\pi/2,2\pi)$. See Fig.~\ref{fig:leaders_in_bottom}. In both cases \ldr{p}
  subdivides $E(\ell)$ into a left region $L$ and a right region $R$. At any
  height ($y$-coordinate), this left region $L$ has width at most one. Hence, if there
  is a point $q \in L$ labeled internally, then its label intersects
  \ldr{p}. Any point $q \in R$ is separated from $S(\ell,u)$ by \ldr{p}. So, if
  there is a label $a \in S(\ell,u)$ that intersects \lbl{q}, then it also
  intersects
  \ldr{p}. 
\end{proof}

\begin{observation}
  \label{obs:closest_leader}
  Let $p \in \P \cap E''(\ell)$ be the point closest to the slab
  $\overline{S(\ell,u)}$, and let $q$ be any point in $\P \cap E''(\ell)$. If
  there is a point $a \in S(\ell,u)$ whose label \lbl{a} intersects \ldr{q},
  then \lbl{a} also intersects~\ldr{p}.
\end{observation}

Observation~\ref{obs:closest_leader} give us that there is only one relevant
point in $E''(\ell)$, namely the point $p$ closest to
$\overline{S(\ell,u)}$. Furthermore, from Lemma~\ref{lem:leaders_in_E''} it
follows that if $p$ exists, then there are no relevant points in $E(\ell)$
labeled internally. Hence, $p$ by itself determines a configuration. We can
then define the universe $\U_{\ell{}E}$ of possible configurations of
$\tilde{B_\ell}$ as follows:
\begin{align*}
  \U_{\ell{}E} &= \{ (\I,\emptyset)    \mid \I \subseteq E(\ell) \text{ and all labels
    in } \Lambda(\I) \text{ are pairwise disjoint} \}~\cup \\
         &~\quad\{ (\emptyset,\{p\}) \mid p \in E''(\ell) \}
\end{align*}

Let $e'(\theta) = 1$ if there is a point in $E''(\ell)$ labeled externally,
and $e'(\theta) = 0$ otherwise (this includes the case in which $E'(\ell) =
\emptyset$). Using that the labels of points in $E(\ell)$ do not contain any
other points, together with Lemma~\ref{lem:delta} we then obtain:

\begin{lemma}
  \label{lem:bottom_configurations}
  The number of configurations of $\tilde{B_\ell}$ is at most
  $O(|\U_{\ell{}E}|)= O(\delta^{e(\theta)} + n^{e'(\theta)})$.
\end{lemma}

\paragraph{Bounding the number of configurations of $\tilde{T_u}$} Analogously
to the bottom influence region $E(\ell)$ in $\tilde{B_\ell}$ we define a
\emph{top influence region} $F(u)$ containing the points from $\tilde{T_u}$
whose label can intersect a label of $S(\ell,u)$, and a region $F'(u)$
containing the points whose leader can intersect a label of the points in
$S(\ell,u)$. We observe that $F(u)$ and $F'(u)$ are symmetric to $E(\ell)$ and
$E'(\ell)$ by mirroring in a line with slope one.\frank{something slightly more
specific would be nice} We thus get
similar results for $F$ and $F'$ as those stated in Lemmas, Corollaries, and
Observations~\ref{lem:sizes_E}-\ref{obs:closest_leader}. So, similarly we
define the universe of configurations $\U_{uF}$ of labelings of
$\tilde{T_\ell}$. We can then summarize our results in the following lemma:
\begin{lemma}
  \label{lem:top_configurations}
  The number of configurations of $\tilde{T_u}$ is at most
  $O(|\U_{uF}|)= O(\delta^{f(\theta)} + n^{f'(\theta)})$, where
  \begin{tabular}{>{\centering\arraybackslash}m{0.33\textwidth}
                  >{\centering\arraybackslash}m{0.34\textwidth}
                  >{\centering\arraybackslash}m{0.3\textwidth}}
  $\displaystyle
  f(\theta) \leq
  \begin{cases}
    0 & \text{if } \theta = 0 \\
    1 & \text{if } \theta \in (0,\pi/4) \\
    0 & \text{if } \theta \in [\pi/4,3\pi/4] \\
    1 & \text{if } \theta \in (3\pi/4,\pi) \\
    0 & \text{if } \theta = \pi \\
  \end{cases}
  $
  &
  $\displaystyle
  f(\theta) \leq
  \begin{cases}
    1 & \text{if } \theta \in (\pi,5\pi/4] \\
    2 & \text{if } \theta \in (5\pi/4,3\pi/2) \\
    1 & \text{if } \theta  = 3\pi/2 \\
    2 & \text{if } \theta \in (3\pi/2,7\pi/4] \\
    3 & \text{if } \theta \in (7\pi/4,2\pi) \\
  \end{cases}
  $
  &
  \includegraphics[page=2]{orientations}
  \end{tabular}
  \noindent and
  \[ f'(\theta) \leq
  \begin{cases}
    1 & \text{if } \theta \in (0,\pi/4) \cup (3\pi/4,2\pi) \\
    0 & \text{otherwise.}
  \end{cases}
\]
\end{lemma}

\subsection{Computing an Optimal Labeling}
\label{sub:Computing_an_Optimal_Labeling}

Let \lab{\ell,u,\C_{\ell{}E},\C_{uF}} denote the maximum number of points in
$S(\ell,u)$ that can be labeled internally, given configurations
$\C_{\ell{}E}=(\I_\ell,\E_\ell)$ and $\C_{uF}=(\I_u,\E_u)$. That is, the
maximum number of points in $S(\ell,u)$ that can be labeled internally assuming
that (i) the points in $\I_\ell \subseteq E(\ell)$ and $\I_u \subseteq F(u)$
are labeled internally, and (ii) the points in $\E_\ell$, $\E_u$, and the
remaining points in $E(\ell)$ and $F(u)$ are labeled externally. An argument
similar to that of Lemma~\ref{lem:combining} then gives us the following
recurrence for \lab{\ell,u,\C_{\ell{}E},\C_{uF}}:
\[
\lab{\ell,u,\C_{\ell{}E},\C_{uF}} = \max~\left\{ \Psi(S(\ell,u)),
  \parbox[c][2cm]{3cm}{
    \begin{equation*}
                                              \max_{\substack{
                                                  p      \in S(\ell,u),\\
                                                  \C_{pE} \in \varrho(  p,\ell,u,\C_{\ell{}E}), \\
                                                  \C_{pF} \in \varsigma(p,\ell,u,\C_{uF}) \\
                                               }}
    \end{equation*}
  }
  \parbox[c][2cm][c]{4cm}{
    \begin{equation*}
      \arraycolsep=1.4pt
    \begin{Bmatrix}
       &\Psi(\tilde{R}_p \cap S(\ell,u))\\
      +&\lab{\ell, p, \C_{\ell{}E}, \C_{pF}}\\
      +&\lab{p,    u, \C_{pE},     \C_{uF}}
    \end{Bmatrix}
    \end{equation*}
  }\right\}
\]
\noindent where $\varrho(p,\ell,u,\C_{\ell{}E})$ and $\varsigma(p,\ell,u,\C_{uF})$ denote the universes
$\U_{pE}$ and $\U_{pF}$ restricted to the sets \emph{compatible} with the
labeling so far, respectively. More formally, we have
\[
  \varrho(p,\ell,u,\C_{\ell{}E}) = \left\{ (\I,\E)
    ~\middle|~~ \parbox[c][3cm][c]{4cm}{
      \begin{align*}
        & (\I,\E) \in \U_{pE},\\
        & \I \supseteq E(p)\cap ((\overline{S(\ell,u)} \cap \tilde{R_p}) \cup \I_\ell) \text{, and }\\
        & \E = \begin{cases}
                    \{\ell\}  & \text{if } \ell \in E''(p) \\
                    \E_\ell    & \text{if } \E_\ell \subset E''(p)\\
                    \emptyset & \text{otherwise.}
                  \end{cases}
      \end{align*}
    }\right\}
\]
%
%
\noindent The function $\varsigma$ is defined analogously.
\maarten [wonders]{Should we add the formula for $\varsigma$ explicitly too? It's just a bit of copy-pasting, but it saves the reader from having to figure out what ``analogously'' means, which I think is not entirely trivial.} \maarten [leaves a dramatic pause, then adds]{Also, it will make our paper look even more impressive.}
Similar to the
previous section we define $\Psi(P)$ as
\[
\Psi(P) =
\begin{cases}
  |P| & \parbox{0.8\textwidth}{if all labels in $\Lambda(P \cup \I_\ell \cup
    \I_u \cup (\overline{S(\ell,u)}\setminus S(\ell,u)))$ are pairwise
    disjoint, and their intersection with the leaders in $\Gamma(\P_X \cup
    \{\ell,u\} \cup \E_\ell \cup \E_u \cup (E(\ell) \setminus \I_\ell) \cup
    (F(u) \setminus \I_u))$ is empty,} \\
  -\infty & \text{otherwise}.
\end{cases}
\]

\paragraph{Computing \lab{\ell,u,\C_{\ell{}E},\C_{uF}}} We again use dynamic
programming. The size of our table is now $O(n^2 |\U_{\ell{}E}| |\U_{uF}|)$. To
compute the value of an entry \lab{\ell,u,\C_{\ell{}E},\C_{uF}}, we maximize
over $O(n |\U_{pE}| |\U_{pF}|)$ other entries. For each such entry, we need to
compute the value of $\Psi(P)$, for some set of points $P$. We can do this in
$O(n\log n)$ time using the algorithm from Section~\ref{sec:left_leaders}
. In
total this yields an $O(n^4 |\U_{\ell{}E}| |\U_{uF}| |\U_{pE}| |\U_{pF}| \log
n)$ time algorithm. Next, we describe how to improve this to $O(n^3 \log n +
n^3 |\U_{\ell{}E}| |\U_{uF}| |\U_{pE}| |\U_{pF}|)$ time by precomputing $\Psi$.

Fix two points $\ell$ and $u$, and let $\Psi'(p,\C_{\ell{}E},\C_{uF}) =
\Psi(\tilde{R}_p \cap S(\ell,u))$, given configurations $\C_{\ell{}E}$ and
$\C_{uF}$. We use a similar approach as in Section~\ref{sec:left_leaders}. We
first compute $\Psi'$ for all points $p$, assuming that no other points above or
below $\overline{S(\ell,u)}$ interfere with $S(\ell,u)$. That is, we compute
all values $\Psi'(p,(\emptyset,\emptyset),(\emptyset,\emptyset))$. This takes $O(n
\log n)$ time using the same algorithm as before. For each of
the remaining pairs of configurations $(C_{\ell{}E},C_{uF})$, we find the
``rightmost'' point $p$ in $\overline{S(\ell,u)}$ such that the label of $p$
conflicts with $C_{\ell{}E}$ or $C_{uF}$. It then follows that
$\Psi'(p',C_{\ell{}E},C_{uF}) = -\infty$ for all $p' \in \tilde{L_p} \cup \{p\}$.

Next, we describe how we can find the ``rightmost'' point that conflicts with
$\C_{\ell{}E}$ in $O(\log n)$ time, after $O(n \log n)$ time preprocessing. We
find the ``rightmost'' point that conflicts with $\C_{uF}$ analogously. It then
follows that we can compute $\Psi'(p,\C_{\ell{}E},\C_{uF})$ for all configurations
and all points $p \in \overline{S(\ell,u)}$ in $O(n \log n +
|\U_{\ell{}E}||\U_{uF}| + (|\U_{\ell{}E}| + |\U_{uF}|)\log n)$ time in
total. \frank{Not sure if we should keep the quotes around rightmost all the
  time.} \martin{I would say in the beginning of the section that the terms rightmost etc. are with respect to $\theta$ and then we can skip the quotes here}
\maarten [helpfully suggests] {Or we can define a different term to mean ``rightmost'', to differentiate from the normal rightmost. Like ``righteous'', for example.}

The ``rightmost'' point that conflicts with a configuration $\C_{\ell{}E} =
(\I,\E)$ conflicts with the set of internally labeled points \I, or the set of
externally labeled points in $\E \cup E(\ell) \setminus \I$. For both these
sets we find the ``rightmost'' point conflicting with it, and return the
``rightmost'' point of those two.

To find the ``rightmost'' point $q$ conflicting with \I, we use the same
procedure as in the previous section. We build a range tree on the corner
points of $\Lambda(\overline{S(\ell,u)})$, and use a priority range query to
find $q_r \in \overline{S(\ell,u)}$ whose label intersects a query label
\lbl{r}, $r \in \I$. We can thus find $q$ in $O(|\I|\log n) = O(\log n)$ time.

If the labels of $S(\ell,u)$ are not contained in $\overline{S(\ell,u)}$ then
we also need to find the ``rightmost'' point whose label intersects a leader in
$Y = \Gamma(\E \cup E(\ell) \setminus \I)$. To do this we need two more data
structures. We preprocess the edges of $Z = \Lambda(\overline{S(\ell,u)})$ to
allow for ray shooting queries with rays of orientation $\theta$. Since all
edges of $Z$ are either horizontal or vertical, and all query rays have the
same orientation, this can be done with a shear transformation and two standard
one-dimensional interval or segment trees $\T'$ (one for the horizontal edges
of $Z$ and one for the vertical edges of $Z$). We can build these trees in $O(n
\log n)$ time~\cite{bkos-cgaa-00}. This allows us to find the ``rightmost''
point $q_r$ whose label intersects a given leader \ldr{r} in $O(\log n)$
time. To find the ``rightmost'' point whose label intersects a leader among
\emph{all} leaders in $Y$ we use the following approach. We query $\T'$ with
the leader $\ldr{r} \in Y$ closest to $\overline{S(\ell,u)}$, and find $q_r$
(if it exists). We then find the first leader $\ldr{s}$ that is hit by a
horizontal rightward ray starting in $r$ (see
Fig.~\ref{fig:horizontal_ray_shooting}), and recursively process \ldr{s}. For
any subsequent pair of such leaders $(\ldr{r},\ldr{s})$ we have that point $s$
lies outside of the label \lbl{r}. Since all but one of these points lie in
$E(\ell)$, there are only a constant number of such pairs. Hence, we also need
only a constant number of queries in $\T'$, each of which takes $O(\log n)$
time.

\begin{figure}[tb]
  \centering
  \includegraphics{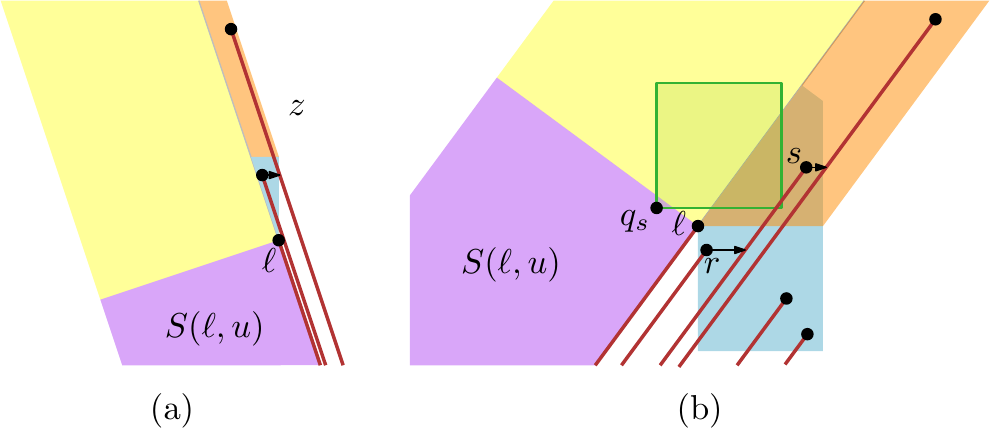}
  \caption{We can find the leaders that can intersect a label from a point in
    $S(\ell,u)$ by a series of horizontal ray shooting queries.}
  \label{fig:horizontal_ray_shooting}
\end{figure}

The only question remaining is how to find the next leader $\ldr{s}$ given
point $r$. To this end we maintain a second data structure. We use a shear
transformation such that the leaders are all vertical. We then build a dynamic
data structure \D for horizontal ray shooting queries among vertical
half-lines. Such a data structure can be build in $O(n \log n)$ time, and
allows for $O(\log n)$ updates and queries. See for
example\cite{giora2009verticalrayshooting}, although much simpler solutions are
possible.\frank{I couldn't find a citation to the easier approach that uses
  that our segments are half-lines, so I just cited the full retroactive
  dictionary. I'm not sure about the exact phrasing to use. Somehow I think it
  is good to stress that it can be done easier than using this full thing, but
  I don't really want to explain the simpler version from scratch, at least not
  for in the conference version.} We can update \D for the next configuration
$\C'_{\ell{}E}$ in $O(\log n)$ time, since only a constant number of points
change from being labeled internally to labeled externally and vice versa.

Hence, after $O(n \log n)$ time preprocessing, we can compute the ``rightmost''
point conflicting with each configuration $\C_{\ell{}E}$ in $O(\log n)$
time. The total time required to compute $\Psi'(p,\C_{\ell{}E},\C_{uF})$, for all
points $p \in S(\ell,u)$, is thus $O(|\U_{\ell{}E}| |\U_{uF}| +
(n + |\U_{\ell{}E}| + |\U_{uF}|)\log n)$.

\newcommand{\myexp}[2]{{#1}_E(\theta) #2 {#1}_F(\theta)}

\newcommand{\minp}{\myexp{s}{\downarrow}}
\newcommand{\maxp}{\myexp{s}{\uparrow}}

\newcommand{\minpp}{\myexp{s'}{\downarrow}}
\newcommand{\maxpp}{\myexp{s'}{\uparrow}}

\newcommand{\sump}{\myexp{s}{+}}
\newcommand{\sumpp}{\myexp{s'}{+}}

After precomputing all values of $f$, the dynamic programming takes
$O(n^3|\U_{\ell{}E}| |\U_{uF}||\U_{pE}| |\U_{pF}|)$ time. We thus obtain a
total running time of $O(n^3 \log n + n^3 |\U_{\ell{}E}| |\U_{uF}| |\U_{pE}|
|\U_{pF}|))$. Using that $|\U_{\ell{}E}|$ and $|\U_{pE}|$ are both at most
$O(n^{e'(\theta)} + \delta^{e(\theta)})$
(Lemma~\ref{lem:bottom_configurations}), and that $|\U_{uF}|$ and $|\U_{pF}|$
are at most $O(n^{f'(\theta)} + \delta^{f(\theta)})$
(Lemma~\ref{lem:top_configurations}) we can rewrite this to $O(n^3(\log n +
\iota(n,\delta,\theta)))$, where $\iota(n,\delta,\theta)$ is a term that models
how much the subproblems can influence each other. We have
%
\begin{alignat*}{19}
 \iota(n,\delta,\theta) =
                        & && n^{2e'(\theta)+2f'(\theta)}&&
                        &&+&& n^{2e'(\theta)+f'(\theta)} &&\delta^{f(\theta)}
                        &&+&& n^{e'(\theta) + 2f'(\theta)}&&\delta^{e(\theta)}\\
                         &+&& n^{e'(\theta)+f'(\theta)}  &&\delta^{e(\theta)+f(\theta)}
                        &&+&& n^{2e'(\theta)}           &&\delta^{2f(\theta)}
                        &&+&& n^{2f'(\theta)}           &&\delta^{2e(\theta)}\\
                         &+&& n^{e'(\theta)}            && \delta^{e(\theta)+2f(\theta)}
                        &&+&& n^{f'(\theta)}            &&\delta^{2e(\theta)+f(\theta)}
                        &&+&&                         &&\delta^{2e(\theta)+2f(\theta)}.
\end{alignat*}
For $\delta = O(n)$, this gives us a worst case running time varying between
$O(n^3\log n)$ and $O(n^{15})$. We conclude:

\begin{proposition}
  Given a set \P of $n$ points and an angle $\theta$, we can compute a labeling
  of \P that maximizes the number of internal labels in $O(n^3(\log n +
  \iota(n,\delta,\theta)))$ time and $O(n^2\sqrt{\iota(n,\delta,\theta)})$
  space, where $\delta = \min \{n, 1/d_{\min}\}$ for the minimum distance
  $d_{\min}$ in \P, and $\iota(n,\delta,\theta)$ models how much subproblems
  can influence each other.
\end{proposition}

\subsection{An Improved Bound on the Number of Configurations}
\label{sub:An_Improved_Bound_on_the_Number_of_Configurations}

The analysis above, together with the fact that $e(\theta)$ and $f(\theta)$
are both at most three, gives us a worst case running time of $O(n^{15})$. We
now study the situation a bit more carefully, and show that the number of
interesting configurations is much smaller. More specifically, that we can
replace $e(\theta)$ and $S_F(\theta)$ by a quantities $e^*(\theta)$ and
$f^*(\theta)$ that are both at most one. This significantly improves the
running time of our algorithm.

We start by observing that for some points $q$ in $E(\ell)$, the subproblem
$S(\ell,u)$ does not have a labeling compatible with $q$, irrespective of
whether $q$ is labeled internally or externally. Let $Q(\ell)$ denote the set
of such points. It follows that we can restrict ourselves to labelings, and
thus configurations, that do not contain points from $Q(\ell)$. Let
$\U^*_{\ell{}E} = \{ (\I,\E) \mid (\I,\E) \in \U_{\ell{}E} \land \I \subseteq
Y(\ell) \}$ denote this subset of configurations, where $Y(\ell) = E(\ell)
\setminus Q(\ell)$.

\begin{lemma}
  \label{lem:CE_at_most_two}
  For every configuration $(\I,\E) \in \U^*_{\ell{}E}$, the set \I has size at
  most two.
\end{lemma}

\begin{proof}
  By Corollary~\ref{cor:at_most_KE_labeled_internally} there are at most
  $e(\theta)$ points in $E(\ell)$ that can be labeled internally
  simultaneously. Hence $|\I| \leq e(\ell)$. For $\theta \in [0,3\pi/2] \cup
  [7\pi/8,2\pi]$ we have $e(\theta) \leq 2$, and thus the lemma follows
  immediately. For $\theta \in (7\pi/8,2\pi)$ we have $e(\theta) \leq 3$. We
  prove this remaining case by contradiction.

  \begin{figure}[tb]
    \centering
    \includegraphics{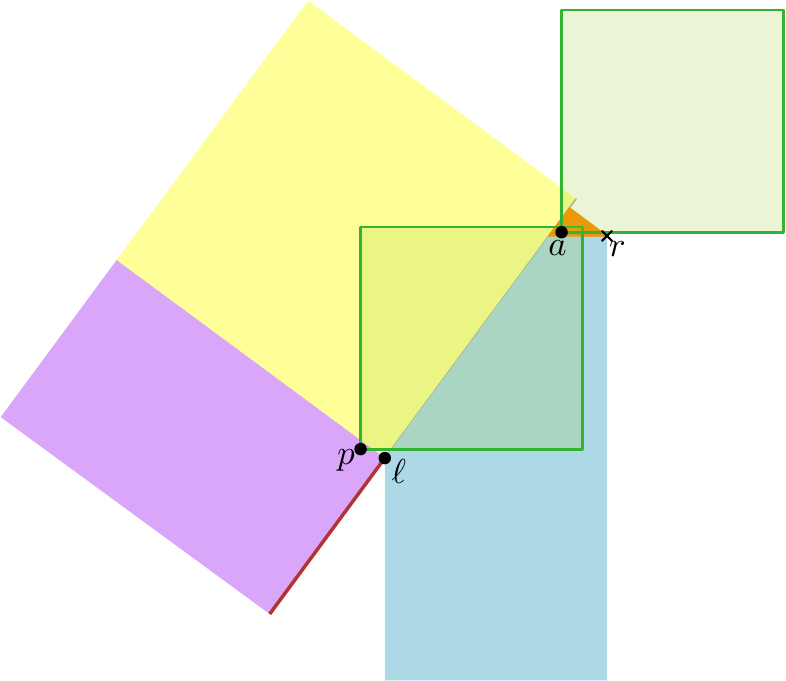}
    \caption{$S(\ell,u)$ is incompatible with any point $a \in \tilde{B_\ell}
      \cap \tilde{L_r} \cap T_r$ (the orange region). }
    \label{fig:ce_at_most_two}
  \end{figure}

  Assume that $\theta \in (7\pi/8,2\pi)$, and that $\I = \{a,b,c\}$. Since the
  region $E(\ell)$ has size $1 \times 3$, it follows that one of these points,
  say point $a$, lies in the triangular region $\tilde{B_\ell} \cap \tilde{L_r}
  \cap T_r$, with $r = (\ell_x+1,\ell_y+1)$. See
  Fig.~\ref{fig:ce_at_most_two}. Let $p$ be a point in $S(\ell,u)$ whose label
  intersects \lbl{a}. Using that $\theta \in (7\pi/8,2\pi)$ it follows that $p$
  is to the bottom-left of $a$. Since the labels are to the top-right of a
  point, we then obtain that $a \in \lbl{p}$. \frank{Actually, this is already
    a contradiction, since we assumed that we removed all those points. If we
    use that the proof actually already works for any configuration in
    $\U_{\ell{}E}$ (rather than only the configs in $\U^*_{\ell{}E})$}
  Therefore, the leader \ldr{a} also intersects \lbl{p}. So, both the label and
  the leader of $a$ interfere with $S(\ell,u)$, and thus $a \not\in X(\ell)$
  and thus also not in $\I \subseteq X(\ell)$. Contradiction.
\end{proof}

\begin{lemma}
  \label{lem:at_most_one_intersecting_label}
  Let $p$ be a point in $S(\ell,u)$, and let $(\I,\E) \in \U^*_{\ell{}E}$. The
  label of $p$ intersects at most one label \lbl{q} of a point $q \in \I$.
\end{lemma}

\begin{proof}
  When $|\I| \leq 1$, the lemma is trivially true. By
  Lemma~\ref{lem:CE_at_most_two}, we otherwise have $|\I| \leq 2$ and $\theta
  \in (0,\pi/4) \cup (3\pi/2,2\pi)$. Let $\I = \{a,b\}$. For the case $\theta
  \in (0,\pi/4)$, assume w.l.o.g. that $b$ is the leftmost point. Since $a$ and
  $b$ are both in \I, their labels are disjoint, and thus we have $b_x + 1 <
  a_x$ (See Fig.~\ref{fig:unique_intersection}(a)). Furthermore, we have $p_y
  \geq \ell_y \geq b_y$. Since \lbl{p} intersects \lbl{b}, but $p
  \not\in\lbl{b}$ this means $p_x < b_x$. Since \lbl{p} has width one, it thus
  cannot intersect \lbl{a}.

  \begin{figure}[h]
    \centering
    \includegraphics{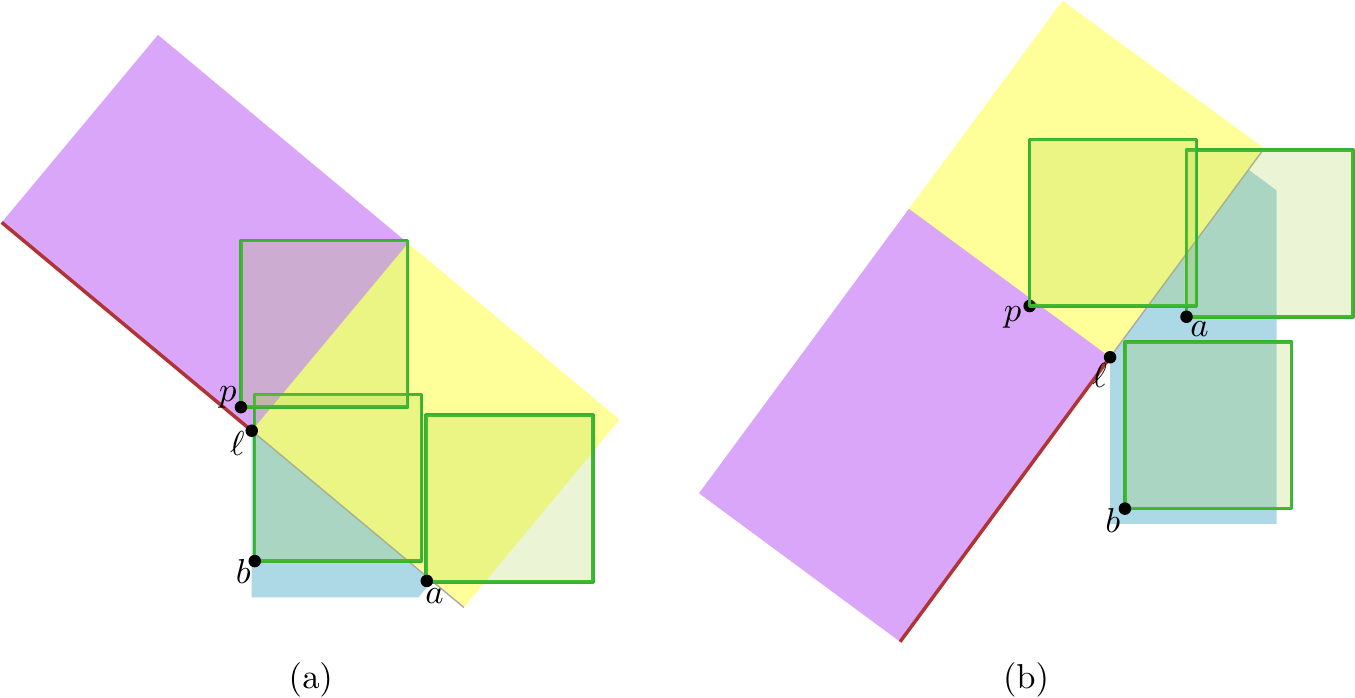}
    \caption{Point $p \in S(\ell,u)$ can intersect only one label of a point in
    $E(\ell)$. \frank{maybe include for which $\theta$ values this is.}}
    \label{fig:unique_intersection}
  \end{figure}

  For the case $\theta \in (3\pi/2,2\pi)$ we use a similar argument. Assume
  w.l.o.g. that $a$ is the topmost point, and thus $a_y > b_y + 1$. See
  Fig.~\ref{fig:unique_intersection}(b). Since $p \in S(\ell,u)$, \lbl{p}
  intersects \lbl{a}, and $a \not\in\lbl{p}$ we have that $p$ is to the
  top-left of $a$. And thus, $p_y > a_y > b_y + 1$. The label of $b$ has height
  at most one, and thus cannot intersect \lbl{p}.
\end{proof}

\begin{lemma}
  \label{lem:b_determines_a}
  Let $\C = (\{a,b\},\E) \in \U^*_{\ell{}E}$, where $a$ is to the right of $b$ if
  $\theta \in (0,\pi/4)$, and to the top of $b$ if $\theta \in
  (3\pi/4,2\pi)$. Point $b$ uniquely determines $a$.
\end{lemma}

\begin{proof}
  We prove this by contradiction. Assume that there is another configuration
  $\C' = (\{a',b\},\E') \in \U^*_{\ell{}E}$, such that $a' \in Y(\ell)$ is to
  the right of $b$ in case $\theta \in (0,\pi/4)$, or above $b$ in case $\theta
  \in (3\pi/4,2\pi)$.

  The points $a$ and $b$ can influence the labeling of $S(\ell,u)$. Hence,
  their labels intersect with labels of points in $S(\ell,u)$. By
  Lemma~\ref{lem:at_most_one_intersecting_label} \lbl{a} and \lbl{b} cannot
  both intersect the same label of a point in $S(\ell,u)$. So, let $p$ be the
  point whose label intersects \lbl{a}, and let $q$ be the point whose label
  intersects \lbl{b}. See Fig.~\ref{fig:b_determines_a}.

  \begin{figure}[tb]
    \centering
    \includegraphics{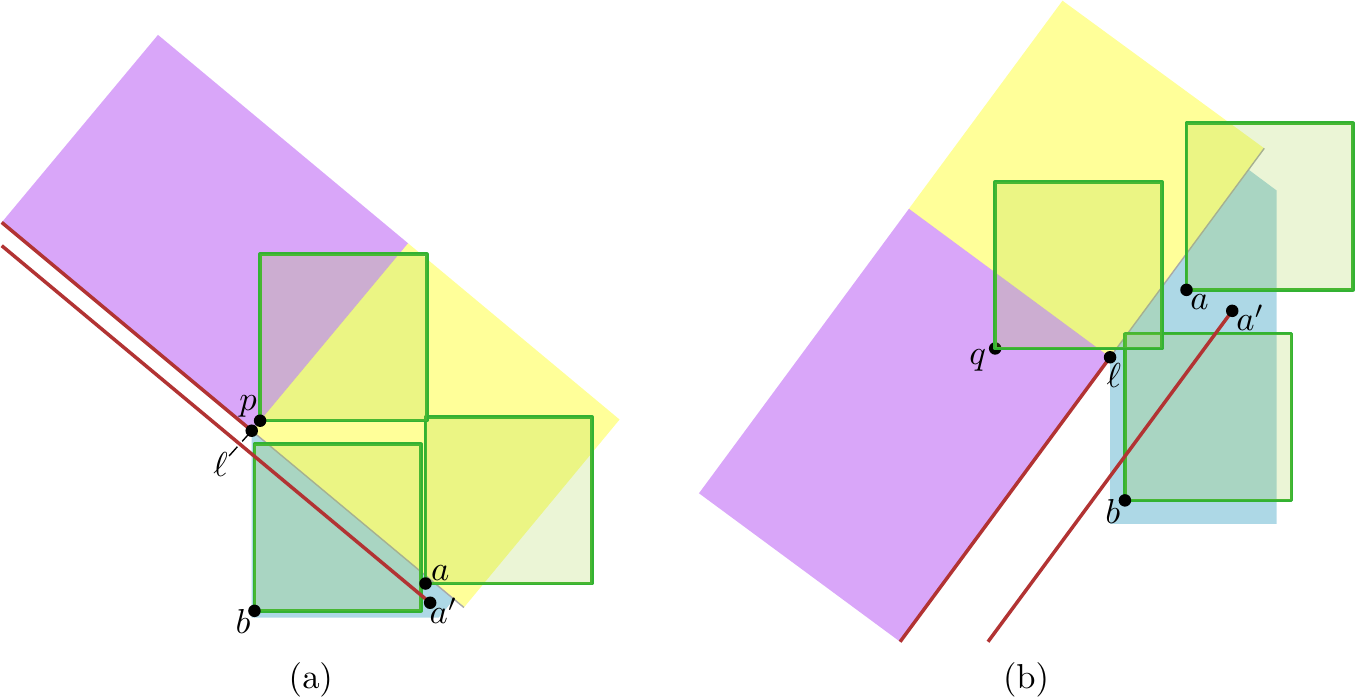}
    \caption{The points $a, a'$, and $b$ in $E(\ell)$ and the label(s) in
      $S(\ell,u)$ they intersect. (Fig (a) is not on scale.) \frank{same here}}
    \label{fig:b_determines_a}
  \end{figure}

  We start with the case $\theta \in (0,\pi/4)$. Point $a$ is to the right of
  $b$, and \lbl{a} and \lbl{b} are disjoint. Hence, $a_x > b_x + 1$. Since
  \lbl{p} intersects \lbl{a}, and $p \not\in\lbl{b}$, it follows that $p$ lies
  above \lbl{p}. Hence $p_y > b_y + 1$. Again using that \lbl{p} intersects
  \lbl{a}, it then also follows that $a_y > b_y$. Using the same argument we
  get $a'_y > b_y$. Now in the configuration \C, point $a'$ is labeled
  externally. However, this means its leader intersects \lbl{b}. Hence, $\C
  \not\in \U^*_{\ell{}E}$. Contradiction.

  For the case $\theta \in (3\pi/2,2\pi)$ we have that $b_y+1 < a_y$. As in the
  proof of Lemma~\ref{lem:CE_at_most_two} we have that $a_y < \ell_y + 1$, and
  hence $b_y < \ell_y$. Using that \lbl{q} does not intersect \ldr{\ell}, we
  have $q_y > \ell_y$. Since $q \not\in \lbl{b}$ it follows that $q_x <
  b_x$. In configuration $\C'$, \lbl{a'} and \lbl{b} are also pairwise
  disjoint. It then follows that $a'$ is to the top-right of $b$. In the
  configuration $\C=(\{a,b\},\E)$, point $a'$ is labeled externally. However,
  this means its leader intersects \lbl{b}. Hence, $\C \not\in
  \U^*_{\ell{}E}$. Contradiction.
\end{proof}

\begin{corollary}
  \label{cor:point_set_choices}
  The number of configurations in $\U^*_{\ell{}E}$ is at most
  $O(n^{e'(\theta)} + \delta^{e^*(\theta)})$, where $e^*(\theta) =
  \min\{1,e(\theta)\}$.
\end{corollary}

\begin{figure}[tbh]
  \centering
  \includegraphics{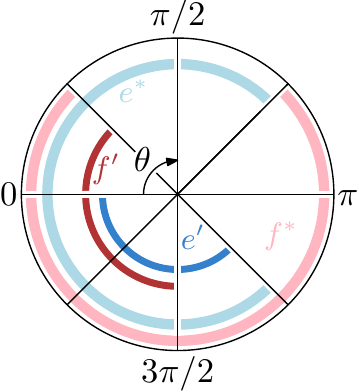}
  \caption{A depiction of the upper bounds on $e^*, e', s^*F,$ and $f'$
    as a function of $\theta$.}
  \label{fig:configurations}
\end{figure}

We can again use a symmetric argument for the number of configurations in
$\tilde{T_u}$. Figure~\ref{fig:configurations} gives a graphical summary of
these results. We now redefine $\iota(n,\delta,\theta)$ as

\begin{alignat*}{19}
 \iota(n,\delta,\theta) =
                        & && n^{2e'(\theta)+2f'(\theta)}&&
                        &&+&& n^{2e'(\theta)+f'(\theta)} &&\delta^{f^*(\theta)}
                        &&+&& n^{e'(\theta) + 2f'(\theta)}&&\delta^{e^*(\theta)}\\
                         &+&& n^{e'(\theta)+f'(\theta)}  &&\delta^{e^*(\theta)+f^*(\theta)}
                        &&+&& n^{2e'(\theta)}           &&\delta^{2f^*(\theta)}
                        &&+&& n^{2f'(\theta)}           &&\delta^{2e^*(\theta)}\\
                         &+&& n^{e'(\theta)}            && \delta^{e^*(\theta)+2f^*(\theta)}
                        &&+&& n^{f'(\theta)}            &&\delta^{2e^*(\theta)+f^*(\theta)}
                        &&+&&                         &&\delta^{2e^*(\theta)+2f^*(\theta)}.
\end{alignat*}

\noindent and obtain the following result, which is at most $O(n^7)$ for
$\delta = O(n)$:

\begin{theorem}
  \label{thm:other_leaders}
  Given a set \P of $n$ points and an angle $\theta$, we can compute a labeling
  of \P that maximizes the number of internal labels in $O(n^3(\log n +
  \iota(n,\delta,\theta)))$ time and $O(n^2\sqrt{\iota(n,\delta,\theta)})$
  space, where $\delta = \min \{n, 1/d_{\min}\}$ for the minimum distance
  $d_{\min}$ in \P, and $\iota(n,\delta,\theta)$ models how much subproblems
  can influence each other. More formally,
\begin{alignat*}{18}
 \iota(n,\delta,\theta) =
                        & && n^{2e'(\theta)+2f'(\theta)}
                        &&+&& n^{2e'(\theta)+f'(\theta)}\delta^{f^*(\theta)}
                        &&+&& n^{e'(\theta) + 2f'(\theta)}\delta^{e^*(\theta)}\\
                         &+&& n^{e'(\theta)+f'(\theta)}  \delta^{e^*(\theta)+f^*(\theta)}
                        &&+&& n^{2e'(\theta)}           \delta^{2f^*(\theta)}
                        &&+&& n^{2f'(\theta)}           \delta^{2e^*(\theta)}\\
                         &+&& n^{e'(\theta)}            \delta^{e^*(\theta)+2f^*(\theta)}
                        &&+&& n^{f'(\theta)}            \delta^{2e^*(\theta)+f^*(\theta)}
                        &&+&&                         \delta^{2e^*(\theta)+2f^*(\theta)},
\end{alignat*}%
\noindent
where $e^*(\theta),f^*(\theta), e'(\theta)$, and $f'(\theta)$ are all at most
one.
\end{theorem}

\section{Extensions}\label{sec:extensions}

So far, we have considered a stylized version of the question we set out to solve. In this section we discuss how our solution may be adapted and extended, depending on the exact requirements of the application. 

\subsection {Optimizing the direction}
\label{sub:optimize_directions}

Rather than fixing the direction for the leaders in advance, we may be willing to let the algorithm specify the optimal orientation that maximizes the number of points that can be labeled internally. Or, perhaps we wish to compute a chart that plots the maximum number of internally labeled points as a function of the leader orientation $\theta$, leaving the final decision to the judgement of the user.

In both scenarios, we need to efficiently iterate over all possible orientations. We adapt our method straightforwardly. 
Let $Q$ be the set of all $4n$ corner points of all potential labels. For every pair
$p,q \in Q$ consider the slope $\theta_{p,q}$ of the line through $p$ and
$q$. All values $\theta_{p,q}$ partition all possible angles into $O(n^2)$
intervals.
For all values $\theta$ in the same interval $J$, any leader $\ldr{p}$
intersects the same set of potential labels, so the optimal set of internal
labels is constant throughout $J$. We compute it separately for each interval.
\maarten {Did we think about whether we can reuse some of the information we already computed? Only two points change order each time. I think it can happen that we need to recompute everything, because the change is low down in the DP order, but can this happen for each of the $n^2$ events? Seems unlikely...}

By applying Theorem~\ref {thm:other_leaders}, we achieve a total of $O(n^2
\cdot n^3(\log n + \iota(n,\delta,\theta))) = O(n^5(\log n +
\iota(n,\delta,\theta)))$ time to compute the optimal labelings for all
orientations, or to optimize the orientation by performing a simple linear
scan.

\subsection{Routing the outer leaders}\label{sec:place_external}

Once the core combinatorial problem of deciding which points have to be labeled
internally is solved, it remains to route the outer leaders and place the external labels.
Since our goal in this paper is to maximize the number of internally labeled points, we are only interested in finding a valid labeling, in which neither labels nor leaders intersect each other.
Let us assume that $\theta \in [0,\pi/2] \cup [3\pi/2, 2\pi]$, i.e., all external labels are oriented to the left. 
The case of labels oriented to the right is symmetric.
We consider the leaders in counterclockwise order around the boundary of \M and place them one by one starting with the topmost leader.
The first label is placed with its lower right corner anchored at the endpoint of its inner leader.
For all subsequent labels we test if the label anchored at the endpoint of the inner leader intersects the previously placed label.
If there is no intersection, we use that label position.
Otherwise, we draw an outer leader extending horizontally to the left starting from the endpoint of the inner leader until the label can be placed without overlap.
Obviously this algorithm takes only linear time.
Figure~\ref{fig:outerleaders} shows the resulting labelings for four different slopes. 
We note that depending on the slope $\theta$ of the inner leaders other methods for routing the outer leaders might yield more pleasing external labelings. 
This is, however, beyond the scope of this paper.

\begin{figure}[tbp]
	\centering
	\subfloat[$\theta=0$]{\includegraphics[page=1]{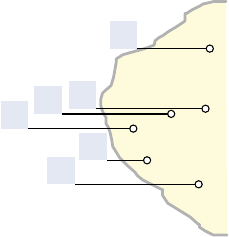}}
	\hfill
	\subfloat[$\theta=\pi/3$]{\includegraphics[page=2]{external-routing}}
	\hfill
	\subfloat[$\theta=7\pi/4$]{\includegraphics[page=3]{external-routing}}
	\hfill
	\subfloat[$\theta=11\pi/6$]{\includegraphics[page=4]{external-routing}}
	\caption{Examples for routing the outer leaders and placing the external labels for different slopes.}
	\label{fig:outerleaders}
\end{figure}

\subsection {Non-square labels}\label{sec:non-square}

Square labels are not very realistic in most map-labeling applications. Their use is justified by the observation that if all labels are homothetic rectangles, we can scale the plane in one dimension to obtain square labels without otherwise changing the problem.
Nonetheless, reality is not quite that simple, for two reasons: firstly, the scaling does alter inter-point distances, so if we wish to parametrize our solution by $d_{\min}$ we need to take this into account. Secondly, in real-world applications, labels may arguably have the same height, but not usually the same width.

If all labels are homothetic rectangles with a height of $1$ and a width of $w$, scaling the plane by a factor $1/w$ in the horizontal direction potentially decreases the closest interpoint distance by the same factor.
Now, the number of points in a unit-area region that do not contain each
other's potential labels is bounded by $\delta w$, immediately yielding a
result of $O(n^3(\log n + \iota(n,\delta w,\theta)))$ using exactly the same approach.

When all labels have equal heights but may have arbitrary widths, we conjecture
that a variation of our approach will still work, but a careful analysis of the
intricacies involved is required. If the labels may also have arbitrary heights
the problem is open. It is unclear if there is a polynomial time solution in
this case.

\subsection {Obstacles}\label{sec:obst}

In this paper, we have considered only abstract point sets to be labeled, using leaders that are allowed to go anywhere, as long as they do not intersect any internal labels.
While this is justified in some applications (e.g., in anatomical drawings, it is common practice to ignore the drawing when placing the leaders, as they are very thin and do not occlude any part of the drawing), in others this may be undesirable (in certain map styles, leaders may be confused for region boundaries or linear features).
As a solution, we may identify a set of polygonal \emph {obstacles} in the map, that cannot be intersected by leaders or internal labels.

In this setting, obviously not every input has a valid labeling: a point that lies inside an obstacle can never be labeled, or obstacles may surround points or force points into impossible configurations in more complex ways. Nonetheless, we can test whether an input has a valid labeling and if so, compute the labeling that maximizes the number of internal labels in polynomial time with our approach.

The main idea is to preprocess the input points in a similar way as in the
beginning of Section~\ref{sec:Other_Directions}. Whenever a point has a
potential leader that intersects an obstacle, it must be labeled internally;
similarly, whenever a point has a potential internal label that intersects an
obstacle, it must be labeled externally. If we include such ``forced'' leaders
or labels into our set of obstacles and apply this approach recursively, we
will either find a contradiction or be left with a set of points whose
potential leaders and potential internal labels do not intersect any obstacle,
and we can apply our existing algorithm on this point set.

The same approach may be used for point sets that are not in general position: if we disallow leaders that pass through other points, they are forced to be labeled internally. Note that this again may result in situations where no valid labeling exists.



\small
\section*{Acknowledgments}
M.L. and F.S. are supported by the Netherlands Organisation for Scientific
Research (NWO) under grant 639.021.123 and 612.001.022, respectively.


\bibliographystyle{abbrv}
\bibliography{bibliography,bibliography2}

\end{document}